%% file: ex_article.tex
\pdfoutput=1
\documentclass[hidelinks,onefignum,onetabnum]{siamart250211}

\input{ex_shared}

\ifpdf
\hypersetup{
  pdftitle={A Kinetic Theory Approach To Ordered Fluids: Direct Simulation Monte Carlo Methods},
  pdfauthor={J.~A.~Carrillo, P.~E.~Farrell, A.~Medaglia, U.~Zerbinati}
}
\fi

\usepackage{caption}
\usepackage{subfig}

\begin{document}

\maketitle

\begin{abstract}
We present a direct simulation Monte Carlo (DSMC) method for the kinetic theory of ordered fluids \revii{recently introduced by the authors~\cite{carrillo2025kinetic}}. The microscopic constituents of these fluids carry an order parameter belonging to a manifold. The method applies Strang splitting to couple a \revii{Monte Carlo treatment of the binary collisions} with a symplectic integration of the mean-field Vlasov transport on the order-parameter manifold. We use the Nanbu--Babovsky approach for Maxwellian collision kernels, and Bird's acceptance--rejection method for general ones. The mean-field Vlasov force is reconstructed from the particle ensemble at each time step. \revii{We also include a Bhatnagar--Gross--Krook (BGK) / Andersen thermostat to control the system temperature.} We formulate the method for a generic order-parameter manifold, and specialize it to two-dimensional calamitic (rod-like) molecules. Numerical tests validate the method. We confirm the relaxation to the Maxwellian distribution predicted by the $H$-theorem, and simulate the mean-field alignment driven by quadratic, Kuramoto, and Onsager interaction potentials. With the Onsager potential we successfully capture the isotropic--nematic phase transition at the correct critical temperature.
\end{abstract}

\begin{keywords}
kinetic theory, ordered fluids, liquid crystals, direct simulation Monte Carlo, order-parameter manifold, mean-field interaction, isotropic--nematic transition, symplectic integration, $H$-theorem
\end{keywords}

\begin{MSCcodes}
82C40, 65C05, 76A15, 82B26, 65P10
\end{MSCcodes}

\section{Introduction}
\label{sec:intro}
Fluids whose microscopic constituents carry a notion of order are ubiquitous in nature and technology. Liquid crystals, whose rod-like (calamitic) molecules tend to align, are the paradigmatic example, with the ordering encoding the orientation of each molecule. Microscopic orderings also arise for ferrofluids, suspensions of elongated colloids, and many biological active fluids. The ordering of the constituents is responsible for the rich macroscopic phenomenology of these materials, such as anisotropic elasticity, anisotropic viscosity, and the emergence of ordered phases through phase transitions.

In our companion work \cite{carrillo2025kinetic}, we developed a kinetic theory for these ordered fluids. The theory encodes the fluid's microstructure in an \emph{order parameter} taking values in a manifold $\mathcal{M}$, which is equipped with an action of the rotation group. Starting from a Hamiltonian description of the microscopic constituents, we derived a Bogoliubov--Born--Green--Kirkwood--Yvon (BBGKY) hierarchy and a Boltzmann--Vlasov kinetic equation for the one-particle distribution function. Using Noether's theorem, we identified the conserved quantities of the binary collisions, and we proved an $H$-theorem to guarantee relaxation to a Maxwellian steady state.

\revam{That work did not address the numerical treatment of the resulting kinetic equation. This equation} is very difficult to simulate because its phase space is high-dimensional, comprising position, translational velocity, order parameter, and conjugate momentum. \revam{Moreover, the motion on the order-parameter manifold is driven by a mean-field force that depends on the unknown distribution itself.}

\revam{Standard deterministic methods, such as finite-volume, discontinuous Galerkin, discrete-velocity, and spectral methods, approximate the distribution function on a phase-space grid \cite{dimarcoPareschi,pareschi2006}, and their cost grows exponentially with the dimension $d$. In our model, this requires resolving simultaneously position, velocity, order parameter, and conjugate momentum. High-order and fast spectral methods can reduce the discretization error or the cost of evaluating the collision operator on a fixed grid, but they do not remove the phase-space scaling. For the dimensions and resolutions considered here, a direct deterministic discretization therefore becomes computationally prohibitive.}

Direct simulation Monte Carlo (DSMC) methods offer a natural alternative. \revam{They have a long history in rarefied gas dynamics \cite{bird,pareschiRussoIntro,pareschi1999implicit,pareschi2001time,pareschi2013interacting} and have recently been coupled with particle-in-cell methods for collisional plasma models \cite{medaglia2026dsmc}. DSMC represents} the distribution function \revam{by} an empirical measure of particles and \revam{models} collisions stochastically. \revam{Typically, the cost of a collision step is linear in the number of particles, and} their convergence rate depends only on the law of large numbers, sidestepping the curse of dimensionality.

\revam{The price paid is statistical noise and a relatively slow convergence with respect to the particle number. Recent work by Borghi \& Pareschi analyzes the convergence of Nanbu-type collision schemes \cite{borghiPareschi2026}; see also the references therein. DSMC thus provides a feasible approach to the present high-dimensional kinetic equation. However, classical DSMC methods must be adapted to account for the motion on the order-parameter manifold and for the self-consistent mean-field force.}

In traditional engineering DSMC, molecules with internal degrees of freedom are typically handled using the phenomenological Borgnakke--Larsen model \cite{borgnakkeLarsen}, which reallocates internal energy during a fraction of collisions by sampling an equilibrium distribution. In contrast, \revam{the kinetic theory considered here \cite{carrillo2025kinetic}} tracks the order parameter and its conjugate momentum \revam{explicitly}, using a collision rule derived directly from microscopic mechanics. Farrell, Russo \& Zerbinati \cite{farrellRusso2024} instead derived a kinetic equation for calamitic constituents without a mean field by closing the hierarchy around an already aligned state, leading to corresponding hydrodynamic models \cite{farrellRusso2024, farrellMalekSoucekZerbinatiLC2026}.

Alternatively, one can build a kinetic description of rod-like particles starting from stochastic, overdamped systems rather than Hamiltonian dynamics. For example, \revii{a matched asymptotic expansion in the occupied volume fraction} for planar Brownian hard needles yields a nonlinear nonlocal equation for the spatial and orientational density. Its spatially homogeneous states successfully recover Onsager's isotropic--nematic transition \revii{\cite{brunaChapmanSchmidtchen2023}}, advancing the coarse-graining of hard-core interactions \cite{brunaChapman2012}. Similarly, Bruna et al.\ \cite{brunaBurgerEspositoSchulz2022} compared mean-field and short-range excluded-volume closures for active Brownian particles oriented on the circle. \revam{These approaches concern different microscopic regimes and do not address the numerical approximation of the Boltzmann--Vlasov equation considered here.}

This paper designs and validates a DSMC scheme for this new kinetic theory of ordered fluids. We formulate the method for a generic order-parameter manifold, and for concreteness we use two-dimensional calamitic molecules (where the order parameter is an angle on the circle $\mathbb{S}^1$) as a running example.

Our numerical scheme uses a Strang splitting \cite{Strang1968} to separate the collision operator from the Vlasov-type transport operator. The collision operator is simulated using the Nanbu--Babovsky DSMC approach \cite{nanbu,babovsky} to strictly preserve the invariants of the binary collision rule. For non-Maxwellian kernels, we combine this with Bird's acceptance--rejection method \cite{bird,koura}, detailed in \cref{sec:kernels}\reviv{, for which no symmetry between the forward and the reverse transition probabilities is needed}. \revii{We also consider the Bhatnagar--Gross--Krook (BGK) relaxation in \cref{sec:bgk}, replacing binary collisions with relaxation toward a local Maxwellian. This also serves as the basis for the thermostat of \cref{sec:thermostat}}.

We highlight two differences compared to classical DSMC for the Boltzmann equation \revam{of neutral particles. Both differences, however, have analogues in other kinetic models.}

First, \revam{instead of undergoing} free streaming \revam{between collisions}, our particles move on the order-parameter manifold driven by a mean-field Vlasov force. \revam{A similar situation arises in plasma models, where the dynamics of charged particles are coupled to an electromagnetic field that must be determined from the corresponding field equations \cite{medaglia2023stochastic,medaglia2026dsmc}. In our setting, the self-consistent force drives the motion on the order-parameter manifold.} We integrate this motion using a symplectic St\"ormer--Verlet scheme \cite{Stormer,Verlet,hairerLubichWannerActa}. A na\"{\i}ve integration would cause the total energy to drift \cite{hairer2006geometric}, obscuring the crucial energy balance (the conversion of interaction energy to kinetic energy) that drives the emergence of order. \revam{Indeed, in our calamitic simulations, even if a symplectic Euler scheme is used the resulting energy drift completely prevents the formation of nematic ordering.}

Second, the mean-field force depends on the unknown distribution itself. \revam{A dependence on $f$ also occurs in quantum kinetic models. In particular, the collision operators for fermionic models used in semiconductor transport, see \cite{markowich1990semiconductor, medaglia2025GNR} and references therein, and for bosonic models, see \cite{kaniadakis1994classical, markowich2005fast} and references therein, contain factors depending on $f$, accounting respectively for Pauli exclusion principle and Bose condensation. In those models this dependence occurs in the collision operator, whereas in our model it enters the transport operator through the mean-field force.} At every time step, we must \revam{therefore} reconstruct \revam{the distribution} from the particle ensemble. Depending on the potential, this requires either \revii{computing a few sample moments} \rev{or reconstructing the full order-parameter density on a mesh of the manifold}\footnote{While unusual for the \revam{classical} Boltzmann equation \revam{of neutral particles}, this is a standard technique in particle-in-cell simulations of \revam{plasma models} \cite[Chapter~2.6]{birdsallLangdon}.}.

We validate our method with a series of numerical tests. First, we confirm that without a mean-field potential, the system relaxes to the Maxwellian distribution predicted by the $H$-theorem \cite{carrillo2025kinetic}. Next, we study the mean-field alignment driven by quadratic, Kuramoto, and Onsager-type potentials. This includes simulating the nucleation of order from a perturbed isotropic state and the collisionless filamentation of phase space.

We also verify that the ordering dynamics are robust to the choice of collision kernel by repeating tests with the physical hard-needle kernel. \rev{Furthermore, we confirm relaxation to the Maxwellian using a collision rule with differing forward and reverse transition probabilities---the specific regime for which the underlying $H$-theorem \cite{carrillo2025kinetic} was designed.}

Finally, we simulate the isotropic--nematic phase transition with the Onsager potential. By varying the bath temperature with \revii{an Andersen} thermostat, we observe a second-order transition at the critical temperature $T_c = 8/\pi$. This precisely matches the critical temperature from a linear stability analysis of the isotropic state (\cref{prop:critical}). The second-order nature of the transition also aligns with the Landau--de Gennes theory of two-dimensional nematics \rev{\cite[Chapter~2]{degennes}} and with the multi-particle collision dynamics simulations of Shendruk \& Yeomans \cite{shendrukYeomans}.

The remainder of the paper is organized as follows. In \cref{sec:model} we recall the kinetic theory of ordered fluids \cite{carrillo2025kinetic} and its specialization to two-dimensional calamitic molecules\rev{, which we develop as a running example throughout the paper}. In \cref{sec:scheme} we present the DSMC scheme: the operator splitting, the symplectic transport step and its energy-conservation properties, the reconstruction of the mean-field force, the collision step, and the thermostat. \Cref{sec:potentials} collects the mean-field potentials we consider. \Cref{sec:experiments} reports the numerical tests, and \cref{sec:conclusion} draws our conclusions.

\section{A kinetic theory of ordered fluids}
\label{sec:model}
We briefly recall the kinetic model \cite{carrillo2025kinetic}. The fluid's microstructure is captured by an order parameter $\nu$ taking values in a smooth, compact \emph{order-parameter manifold} $\mathcal{M}$. The rotation group $SO(d)$ acts on $\mathcal{M}$ via a group action $\mathcal{A}$, where $d$ is the physical dimension. Each microscopic constituent is completely described by its position $\pt{x}$, translational momentum $\vec{p}$, order parameter $\nu$, and conjugate momentum $\vec{\varsigma}$.

Assuming a weak interaction between the ordered constituents, we developed a BBGKY hierarchy to derive a Boltzmann--Vlasov equation \cite[equation~(3.1)]{carrillo2025kinetic} for the one-particle distribution function $f=f(\pt{x},\nu,\vec{p},\vec{\varsigma},t)$:
\begin{equation}
	\label{eq:boltzmann_generic}
	\frac{\partial f}{\partial t} + m^{-1}\vec{p}\cdot\nabla_{\pt{x}}f + \mat{B}(\nu)^{-1}\vec{\varsigma}\cdot\nabla_{\nu}f + \mathcal{V}\cdot\nabla_{\vec{\varsigma}}f = \mathcal{C}[f,f].
\end{equation}
Here, $m$ is the molecular mass, $\mat{B}(\nu)$ is the (symmetric, positive-definite) generalized inertia tensor, $\mathcal{C}$ is the binary collision operator, and $\mathcal{V}(\vec{x},\nu,t)$ is a mean-field Vlasov force acting on the order parameters generated by an interaction potential $\mathcal{W}(\vec{x},\vec{x}_*,\nu,\nu_*)$. \reviv{The interaction potential $\mathcal{W}$ is the potential energy of a pair of constituents, and the mean-field force is the average of $-\nabla_\nu\mathcal{W}$ against the distribution of the partner \cite[equation~(3.36)]{carrillo2025kinetic}.}
The collision operator conserves mass, linear momentum, generalized angular momentum, and energy \cite[Theorem~3.11]{carrillo2025kinetic}. 
Notice that due to the way the generalised inertia matrix $\mat{B}(\nu)$ is defined \cite[eq. 2.18]{carrillo2025kinetic}, the transport part of \eqref{eq:boltzmann_generic}, that is the right-hand side of \eqref{eq:boltzmann_generic}, conserves mass, linear momentum and generalised angular momentum. This together with the conservation properties of the collision operator imply the conservation of mass, linear momentum and generalised angular momentum for the all dynamic \cite{farrellMalekSoucekZerbinatiLC2026}. Furthermore, it satisfies an $H$-theorem \cite[Theorem~3.12]{carrillo2025kinetic}: the only equilibria are the Maxwellian distributions, and spatially homogeneous dynamics will inevitably relax toward them \cite[Theorem~3.13]{carrillo2025kinetic}.

This paper focuses entirely on the spatially homogeneous setting. Because we assume that $f$ is independent of position $\pt{x}$, physical transport is absent and translational velocities only matter during collisions.

\reviv{In this setting the mean-field force is conservative and its} potential is obtained by convolving the pairwise interaction potential with the \reviv{particle density in the order parameter}\revii{, the convolution being defined as}
\begin{equation}\label{eq:conservative_force_generic}
	\mathcal{V}(\nu) = -\nabla_\nu \Phi(\nu), \qquad \Phi(\nu) = (\mathcal{W}\ast\rho)(\nu) \revii{{}= \int_{\mathcal{M}} \mathcal{W}(\nu,\nu_*)\,\rho(\nu_*)\,d\nu_*},
\end{equation}
\revii{where $\rho(\nu,t) = \int f \,d\vec{p}\,d\vec{\varsigma}$ is the \reviv{particle density in the order parameter}.}

\rev{The central conserved quantity is the total energy $\mathcal{E}$. In this homogeneous setting, $\mathcal{E}$ is the sum of the kinetic energy and the mean-field interaction energy:
\begin{equation}\label{eq:totalenergy_generic}
	\mathcal{E} = \underbrace{\int \tfrac{1}{2}\big(m^{-1}|\vec{p}|^2 + \vec{\varsigma}\cdot\mat{B}(\nu)^{-1}\vec{\varsigma}\big)\, f\,d\vec{p}\,d\nu\,d\vec{\varsigma}}_{\text{kinetic energy}} \; + \; \underbrace{\tfrac{1}{2}\iint \mathcal{W}(\nu,\nu_*)\,\rho(\nu)\,\rho(\nu_*)\,d\nu\,d\nu_*}_{\text{interaction energy}},
\end{equation}
where $\rho$ is the \reviv{particle density in the order parameter of \cref{eq:conservative_force_generic}}. 
Provided the interaction potential $\mathcal{W}$ is symmetric, the dynamics strictly conserves $\mathcal{E}$ (\cref{prop:energy}).}
\begin{proposition}[Conservation of the total energy]
\label{prop:energy}
	\reviv{Let the group action $\mathcal{A}$ on the order-parameter manifold be transitive, and let $f$ be a solution of the spatially homogeneous form of \cref{eq:boltzmann_generic} with elastic collisions and with the conservative mean-field force \cref{eq:conservative_force_generic}, where the interaction potential is symmetric, $\mathcal{W}(\nu,\nu_*)=\mathcal{W}(\nu_*,\nu)$. Then the total energy \rev{\cref{eq:totalenergy_generic}} is conserved in time.}
\end{proposition}
\begin{proof}
	The action being transitive, the inertia tensor $\mat{B}$ is constant over $\mathcal{M}$ \cite[Proposition~2.17]{carrillo2025kinetic}, so the kinetic energy density $\psi(\vec{p},\vec{\varsigma})=\tfrac12(m^{-1}|\vec{p}|^2+\vec{\varsigma}\cdot\mat{B}^{-1}\vec{\varsigma})$ does not depend on $\nu$. We differentiate the two terms of $\mathcal{E}$ in time. The function $\psi$ is a collision invariant of the elastic collisions, $\int \psi\, \mathcal{C}[f,f]\,d\vec{p}\,d\nu\,d\vec{\varsigma} = 0$ \cite[Theorem~3.11]{carrillo2025kinetic}, and the collisions do not alter the order parameters, so they affect neither $\rho$ nor the interaction energy. It remains to account for the transport terms of \cref{eq:boltzmann_generic}. The streaming term contributes nothing to the kinetic energy: since $\psi$ does not depend on $\nu$, the integrand $\psi\,\mat{B}^{-1}\vec{\varsigma}\cdot\nabla_\nu f$ is an exact divergence on the closed manifold $\mathcal{M}$, on which the Hamiltonian streaming is incompressible, Hamiltonian flows preserving phase-space volume \cite[Theorem~8.3]{fasanoMarmi}, and its integral vanishes. The force term is integrated by parts in $\vec{\varsigma}$, the boundary term vanishing by the decay of $f$ at large momenta, so that
	\begin{equation*}
		\frac{d}{dt}\int \psi f = \int \nabla_{\vec{\varsigma}}\psi\cdot\mathcal{V}\,f = \int \big(\mat{B}^{-1}\vec{\varsigma}\big)\cdot\mathcal{V}(\nu)\, f = -\int_{\mathcal{M}} \nabla_\nu\Phi(\nu)\cdot\mu(\nu)\,d\nu,
	\end{equation*}
	where $\mu(\nu) = \int \mat{B}^{-1}\vec{\varsigma}\, f \,d\vec{p}\,d\vec{\varsigma}$ is the orientational flux. For the interaction energy, the product rule gives
	\begin{align*}
		\frac{d}{dt}\,\frac{1}{2}\iint \mathcal{W}\,\rho(\nu)\,\rho(\nu_*)
		&= \frac{1}{2}\iint \mathcal{W}\,\partial_t\rho(\nu)\,\rho(\nu_*) + \frac{1}{2}\iint \mathcal{W}\,\rho(\nu)\,\partial_t\rho(\nu_*) \\
		&= \iint \mathcal{W}(\nu,\nu_*)\,\partial_t\rho(\nu)\,\rho(\nu_*)\,d\nu\,d\nu_*,
	\end{align*}
	where the second equality follows by exchanging the names of the integration variables in the second term and using the symmetry of $\mathcal{W}$. Integrating \cref{eq:boltzmann_generic} in $\vec{p}$ and $\vec{\varsigma}$ yields the continuity equation
	\begin{equation*}
		\partial_t \rho + \nabla_\nu\cdot\mu = 0,
	\end{equation*}
	and therefore
	\begin{equation*}
		\iint \mathcal{W}(\nu,\nu_*)\,\partial_t\rho(\nu)\,\rho(\nu_*)\,d\nu\,d\nu_*
		= -\int_{\mathcal{M}} \Phi(\nu)\,\nabla_\nu\cdot\mu(\nu)\,d\nu = \int_{\mathcal{M}} \nabla_\nu\Phi(\nu)\cdot\mu(\nu)\,d\nu,
	\end{equation*}
	the integration by parts on the closed manifold carrying no boundary term. The two contributions cancel, hence $\frac{d}{dt}\mathcal{E}=0$.
\end{proof}

\rev{Throughout the paper we propose a DSMC algorithm for a generic order parameter manifold. For concreteness we consider throughout the case of two-dimensional calamitic molecules, developed in a sequence of \reviv{dedicated paragraphs}.}

\subsection{Two-dimensional calamitic molecules}
\label{sec:calamitic}
\rev{We recall the model of calamitic molecules \cite[Example~2.4]{carrillo2025kinetic}.} These are segments of vanishing girth, uniform length $L$ and mass $m$, whose orientation is described by the unit vector $\vec{\nu}=(\cos\theta,\sin\theta)$. The order-parameter manifold is the circle $\mathbb{S}^1$, parameterized by the angle $\theta\in[-\pi,\pi]$, the generalised inertia reduces to the scalar moment of inertia $I_3$, and the conjugate momentum is proportional to the angular velocity $\omega$. In these coordinates, and with a Maxwellian collision kernel, \cref{eq:boltzmann_generic} becomes, using the parameterization in \cite[Examples~2.4 and~3.19]{carrillo2025kinetic},
\begin{equation}
	\label{eq:Boltzmann2D_calamitic}
	\frac{\partial f}{\partial t} +\omega \partial_{\theta}f + \mathcal{V}\partial_{\omega}f =  \frac{1}{\tau}\int \!\!\!\!\int \!\!\!\!\int \left( f' f_*' - f f_*\right) dv_* d\theta_* d\omega_*,
\end{equation}
where $f=f(\vec{v},\theta,\omega,t)$, $f_*=f(\vec{v}_*,\theta_*,\omega_*,t)$, and $f',f_*'$ are the distributions depending on the post-interaction coordinates, while \reviv{$1/\tau>0$ is the collision frequency, which multiplies the collision operator after a rescaling of the time variable}.
\reviv{The Maxwellian kernel gives every pair of molecules the same collision rate, whatever their relative velocity and orientations, and it is the simplest kernel with which to present the scheme. The physical hard-needle kernel defined in \cite[Example~2.33]{carrillo2025kinetic} is treated in \cref{sec:kernels}, and \cref{sec:hardneedle} shows that the ordering dynamics does not depend on this choice.}
In the generic equation \cref{eq:boltzmann_generic} the conjugate momentum $\vec{\varsigma}$ is a vector of the tangent space of $\mathcal{M}$ at $\nu$, and it is on this vector that the numerical scheme acts. For two-dimensional calamitic molecules the order-parameter manifold is the circle, whose tangent bundle is globally trivial, $T\mathbb{S}^1 \cong \mathbb{S}^1\times\mathbb{R}$. The conjugate momentum is therefore described by a single coordinate in this \revii{setting}, proportional to the angular velocity $\omega$, and we treat it as a scalar throughout. Explicitly, the inertia tensor reduces to the scalar $\mat{B} = I_3\mat{I}$ and $\varsigma = I_3\,\omega$, so that the force appearing in \cref{eq:Boltzmann2D_calamitic} acts on the angular velocity and is therefore the angular \emph{acceleration}, the force on the conjugate momentum divided by $I_3$. For a general manifold the drift and kick maps introduced in \cref{sec:transport} below act on the vector $\vec{\varsigma}$, and it is the vector space structure of the tangent space that makes the kick step well defined.
We complement the equation with the initial condition $f(\vec{v},\theta,\omega,0) = f^0(\vec{v},\theta,\omega)$.

\begin{remark}
	To lighten the notation, in view of the presentation of the numerical scheme, we have introduced some differences with respect to the notation used previously \cite{carrillo2025kinetic}. We have dropped the subscript $1$ in the phase-space coordinates and replaced the subscript $2$ with the subscript $*$. Furthermore, we have dropped the subscript $1$ in the distribution function $f$.
\end{remark}

\section{A direct simulation Monte Carlo \texorpdfstring{\reviv{splitting}}{splitting} scheme}
\label{sec:scheme}
We now present the DSMC \reviv{splitting} scheme for the kinetic equation of ordered fluids. The construction is generic and applies to any order-parameter manifold. The manifold enters through the explicit form of the binary collision rule, through the differential structure needed to compute the mean-field force, which involves the gradient of the interaction potential on $\mathcal{M}$, and through the inertia tensor $\mat{B}(\nu)$. %

\subsection{Operator splitting}
\label{sec:splitting}
We discretize the time interval $[0,t_f]$, with $t_f>0$ the final simulation time, in steps of size $\Delta t>0$, so that $t^n=n\Delta t$, with $n\in\mathbb{N}$. \rev{By $f^n(\vec{v},\nu,\vec{\varsigma})$ we denote an approximation of $f(\vec{v},\nu,\vec{\varsigma},t^n)$,} and we apply a splitting method between the Vlasov-type transport operator and the collisional operator. First we solve the Vlasov-type step $\hat{f}=\mathcal{T}_{\Delta t}(f^n)$
\begin{equation}\label{eq:f_star}
	\left\lbrace
	\begin{aligned}
		&\rev{\frac{\partial \hat{f}}{\partial t} + \mat{B}(\nu)^{-1}\vec{\varsigma}\cdot\nabla_{\nu} \hat{f} + \mathcal{V}\cdot\nabla_{\vec{\varsigma}} \hat{f} = 0}\\
		&\rev{\hat{f}(\vec{v},\nu,\vec{\varsigma},0) = f^n(\vec{v},\nu,\vec{\varsigma})}
	\end{aligned}
	\right.
\end{equation}
and then we solve the collision step $\hathat{f}=\mathcal{Q}_{\Delta t}(\hat{f})$ with initial data given by the solution of the previous step
\begin{equation}\label{eq:f_star_star}
	\left\lbrace
	\begin{aligned}
		&\frac{\partial \hathat{f}}{\partial t}  = \frac{1}{\tau}\mathcal{C}[\hathat{f},\hathat{f}]\\
		&\rev{\hathat{f}(\vec{v},\nu,\vec{\varsigma},0) = \hat{f}(\vec{v},\nu,\vec{\varsigma},\Delta t).}
	\end{aligned}
	\right.
\end{equation}
\revii{The characteristics of \cref{eq:f_star} follow curves on the order-parameter manifold. \Cref{rem:manifold} explains how our discrete transport step stays precisely on this manifold.}

Instead of a simple first-order Lie splitting $f^{n+1}=\mathcal{Q}_{\Delta t}(\mathcal{T}_{\Delta t}(f^n))$, we use a symmetric \emph{Strang} splitting \cite{Strang1968}\revii{, i.e.,}
\begin{equation}\label{eq:strang}
\rev{f^{n+1} =\mathcal{T}_{\Delta t/2}\Big(\mathcal{Q}_{\Delta t}\big(\mathcal{T}_{\Delta t/2}(f^n)\big)\Big).}
\end{equation}
Here, the transport operator is applied over two half-steps that bracket a full collision step. Strang splitting achieves second-order accuracy when its substeps are exact. In our case, the transport substep $\mathcal{T}$ moves particles under the mean-field force via a second-order St\"ormer--Verlet scheme (\cref{sec:transport}). The collisional substep $\mathcal{Q}$, however, uses the Nanbu--Babovsky DSMC method \cite{nanbu,babovsky}\reviv{, or Bird's method \cite{bird} for a general kernel, and both are} only first-order accurate in time \cite{babovskyIllner}.
Consequently, our overall method is strictly first-order accurate. \rev{To regain second-order accuracy, one could substitute the collision substep with time-relaxed Monte Carlo schemes based on a truncated Wild expansion \cite{pareschi2001time}, though we do not pursue that here. Because this is the first Monte Carlo simulation of the kinetic theory \cite{carrillo2025kinetic}, we use a classical DSMC methodology in the spatially homogeneous setting to focus the presentation on the new issues that arise and their possible solutions.} We nevertheless keep the symmetric structure of \cref{eq:strang} because it is essential for total energy conservation, as discussed in \cref{sec:transport}. For higher-order splittings, see \cite{McLachlanQuispel2002} for a broad survey, and \cite{dimarco2015numerical,medaglia2023stochastic,medaglia2026dsmc} for applications in plasma.

\subsection{Particle approximation}
\label{sec:particles}
We approximate the distribution function with a sample of $N$ particles. At time $t$, each particle is identified by its velocity $\vec{v}_i(t)$, order parameter \rev{$\nu_i(t)$}, and conjugate momentum \rev{$\vec{\varsigma}_i(t)$}\reviv{, and we write $\vec{v}^n_i$, $\nu^n_i$, $\vec{\varsigma}^n_i$ for its state at time $t^n$}. As in \cite[equation~(3.61)]{carrillo2025kinetic}, this yields the empirical measure
\begin{equation}\label{eq:empirical}
\rev{f(\vec{v},\nu,\vec{\varsigma},t) \approx f^{N}(\vec{v},\nu,\vec{\varsigma},t) = \frac{1}{N}\sum_{i=1}^{N} \delta(\vec{v}-\vec{v}_i(t))\otimes \delta(\nu-\nu_i(t))\otimes\delta(\vec{\varsigma}-\vec{\varsigma}_i(t)).}
\end{equation}

\rev{To capture local macroscopic fields, we partition the order-parameter manifold into a mesh $\{I_l\}_{l=1}^{N_{\mathcal{M}}}$. On this mesh we accumulate the local moments of the sample: mass, mean velocity, mean conjugate momentum, and temperature.} \reviv{Each particle carries the weight $1/N$ of \cref{eq:empirical} and deposits it in the cell containing its order parameter.}
\revii{Following the equipartition theorem, a cell's temperature is defined as twice the mean kinetic energy of its particles (measured relative to the cell's mean velocity and conjugate momentum) divided by the $N_{\mathrm{dof}}$ kinetic degrees of freedom.}

\paragraph{Two-dimensional calamitic molecules}
\rev{The manifold is the circle, and we mesh the chart $[-\pi,+\pi]$ with $N_\theta$ cells $\{I_l\}_{l=1}^{N_\theta}$ of size $\Delta \theta=2\pi/N_\theta$.} \reviv{In this chart the state of particle $i$ at time $t^n$ is $(\vec{v}^n_i,\theta^n_i,\omega^n_i)$, with $\theta^n_i$ its angle and $\omega^n_i=\varsigma^n_i/I_3$ its angular velocity.} \rev{The $\theta$-local moments of the distribution then read}
\begin{equation}\label{eq:moments}
	\begin{split}
		\rho^n_{\theta,l} &= \frac{m}{N} \sum_{i=1}^N \reviv{\eta_{l}}(\theta^n_i), \qquad
		\Omega^n_{l} = \frac{m}{N\rho^n_{\theta,l}} \sum_{i=1}^N \omega^n_i\, \reviv{\eta_{l}}(\theta^n_i), \qquad
		\vec{U}^n_{l} = \frac{m}{N\rho^n_{\theta,l}} \sum_{i=1}^N \vec{v}^n_i\, \reviv{\eta_{l}}(\theta^n_i), \\
		T^n_l &= \frac{2}{N_{\mathrm{dof}}} \frac{m}{N\rho^n_{\theta,l}} \sum_{i=1}^N \frac{1}{2} \left( m|\vec{v}^n_i-\vec{U}^n_{l}|^2 + I_3 (\omega^n_i-\Omega^n_l)^2 \right) \reviv{\eta_{l}}(\theta^n_i), \qquad \reviv{l}=1, \dots, N_\theta.
	\end{split}
\end{equation}
\reviv{Here $\eta_{l}$ is the indicator of the cell $I_l$, so that \cref{eq:moments} are cell averages of the sample. Smoother shape functions produce smoother moments.} The factor $2/N_{\mathrm{dof}}$ in the local temperature $T^n_l$ reflects the $N_{\mathrm{dof}}$ kinetic degrees of freedom per molecule, enforcing the equipartition theorem. For a two-dimensional calamitic molecule, the two translational and one rotational velocities yield $N_{\mathrm{dof}}=3$.

\subsection{The transport step: symplectic integration and energy conservation}
\label{sec:transport}
The Vlasov transport step $\mathcal{T}_{\Delta t}(\cdot)$ is the push-forward of the distribution along the characteristics of \cref{eq:f_star}. At the particle level, this describes motion on the order-parameter manifold governed by:
\begin{equation}
	\label{eq:ODEs_2}
	\rev{\frac{d \nu_i}{d t}  = \mat{B}^{-1}\vec{\varsigma}_i,\qquad \frac{d \vec{\varsigma}_i}{d t}  = \mathcal{V}(\nu_i).}
\end{equation}
In this work, all considered mean-field potentials generate a Vlasov force \rev{$\mathcal{V}=\mathcal{V}(\nu)$ that depends solely on the order parameter}. We integrate this flow by splitting it into two exactly solvable components: a \emph{drift} $\mathcal{D}_{\delta}$ and a \emph{kick} $\mathcal{K}_{\delta}$. \rev{The drift advances the order parameter along a free trajectory on $\mathcal{M}$ while holding the conjugate momentum constant. Conversely, the kick updates the conjugate momentum while freezing the order parameter. As noted in \reviv{\cref{sec:calamitic}}, the kick is well defined because the tangent space is a vector space.}

We realize the transport substep through the symmetric St\"ormer--Verlet composition:
\begin{equation}\label{eq:verlet}
	\mathcal{T}_{\Delta t} \approx \mathcal{D}_{\Delta t/2}\,\mathcal{K}_{\Delta t}\,\mathcal{D}_{\Delta t/2}.
\end{equation}
Combining this with the Strang splitting \cref{eq:strang}, a full time step becomes:
\begin{equation}\label{eq:fullstep}
	f^{n+1} = \mathcal{D}_{\Delta t/4}\,\mathcal{K}_{\Delta t/2}\,\mathcal{D}_{\Delta t/4}\,\circ\,\mathcal{Q}_{\Delta t}\,\circ\,\mathcal{D}_{\Delta t/4}\,\mathcal{K}_{\Delta t/2}\,\mathcal{D}_{\Delta t/4}\,(f^n).
\end{equation}

\begin{remark}
\label[remark]{rem:manifold}
	\revii{Both halves of the transport substep keep the state on the manifold by construction. The kick acts on the conjugate momentum alone, at frozen order parameter, so it remains in the vector space $T_\nu\mathcal{M}$, as observed in \reviv{\cref{sec:calamitic}}. The drift is the free motion generated by the kinetic energy, whose solution is the geodesic of the metric defined by the inertia tensor, so the order parameter travels along a curve of $\mathcal{M}$ and the conjugate momentum is carried along it.
\end{remark}
\begin{remark}
Two implementations of the drift are available, with differing cost. One integrates the free motion in a chart of $\mathcal{M}$, which is the route we follow. For the circle the drift is the rotation of \cref{eq:driftkick} below, the exponential map of $\mathbb{S}^1$ written in the angular chart, so the substep is exact and the reduction modulo $2\pi$ is a change of representative in the chart and no projection is needed. The price of this route is that it requires an atlas, which is trivial for the circle but might be much more complex for a generic manifold. The alternative embeds $\mathcal{M}$ in a Euclidean space, takes an unconstrained step, and retracts the result onto $\mathcal{M}$. This route needs no atlas, but it costs a closest-point projection per particle and per time step, and the projection destroys the symplecticity of the substep unless the constraint is enforced by a method of RATTLE type \cite[Section~VII.1]{hairer2006geometric}. These two approaches correspond to the local-coordinate and projection methods for differential equations on manifolds described in \cite[Sections~IV.4 and~IV.5]{hairer2006geometric}.}
\end{remark}

All the mean-field forces considered in this work are conservative, meaning that they are generated by a potential\reviv{, as in \cref{eq:conservative_force_generic}}.

\paragraph{Two-dimensional calamitic molecules}
\rev{In the calamitic chart the equations of motion \cref{eq:ODEs_2} read $\dot{\theta}_i=\omega_i$ and $\dot{\omega}_i=\mathcal{V}(\theta_i)$, and the drift and the kick are explicit,}
\begin{equation}\label{eq:driftkick}
	\mathcal{D}_{\delta}:\ \theta_i \mapsto \theta_i + \omega_i\,\delta \!\!\pmod{2\pi}, \qquad
	\mathcal{K}_{\delta}:\ \omega_i \mapsto \omega_i + \mathcal{V}(\theta_i)\,\delta .
\end{equation}
\revii{The reduction modulo $2\pi$ is the periodicity of the angular chart and not a projection back onto the manifold, see \cref{rem:manifold}.}
\rev{The force of \cref{eq:conservative_force_generic} is the torque divided by the moment of inertia,}
\begin{equation}\label{eq:conservative_force}
	\mathcal{V}(\theta) = -I_3^{-1}\,\partial_\theta \Phi(\theta), \qquad \Phi(\theta) = (\mathcal{W}\ast\rho_\theta)(\theta) \revii{{}= \int_{-\pi}^{\pi}\mathcal{W}(\theta,\theta_*)\,\rho_\theta(\theta_*)\,d\theta_*},
\end{equation}
where $\mathcal{W}$ is the pairwise interaction potential and $\rho_\theta$ the orientational density. \revii{The chart being periodic, a translation-invariant kernel $\mathcal{W}(\theta,\theta_*)=\mathcal{W}(\theta-\theta_*)$ has to be $2\pi$-periodic for the convolution to be well defined on the circle, which is the case for all the kernels of \cref{sec:potentials}.} \rev{In all our tests $m=1$ and $I_3=1$, so that we simply write $\mathcal{V}=-\partial_\theta\Phi$. In the tests in which the length of the molecules enters the dynamics through the mean-field potential or the collision kernel we take $L=\sqrt{12}$, so that $I_3=mL^2/12=1$.}
\rev{The spatially homogeneous dynamics conserves the total energy \cref{eq:totalenergy_generic}, and the symplectic integration allows the fully discrete scheme to approximate this conservation, up to a small residual error.} In the tests of \cref{sec:onsager} this residual error stays below $0.6\%$ of the total energy over the full simulation horizon, both with the collisions switched on ($0.6\%$) and off ($0.4\%$). \rev{The residual conservation error has two sources. The mean-field force is reconstructed from the particle sample at every time step (\cref{sec:meanfieldforce}), so the particles do not follow the flow of a fixed Hamiltonian. Moreover, even for a fixed force, a symplectic integrator conserves a modified energy rather than the energy itself.}

The two-dimensional calamitic molecules of \cref{eq:Boltzmann2D_calamitic} satisfy the hypothesis of \cref{prop:energy}, the action of $SO(2)$ on the circle being transitive, and in their chart the conserved energy \cref{eq:totalenergy_generic} reads
\begin{equation}\label{eq:totalenergy}
	\mathcal{E} = \underbrace{\int \tfrac{1}{2}\big(|\vec{v}|^2 + I_3\,\omega^2\big)\, f\,d\vec{v}\,d\theta\,d\omega}_{\text{kinetic energy}} \; + \; \underbrace{\tfrac{1}{2}\iint \mathcal{W}(\theta,\theta_*)\,\rho_\theta(\theta)\,\rho_\theta(\theta_*)\,d\theta\,d\theta_*}_{\text{interaction energy}},
\end{equation}
where we have used $m=1$, and the orientational flux is the scalar $\mathcal{J}(\theta)=\int \omega f\,d\vec{v}\,d\omega$.

At the discrete level the two substeps respect the two terms of \cref{eq:totalenergy} separately. \rev{The Nanbu--Babovsky collision step introduced below in \cref{sec:collision}, \revii{with elastic collisions}, conserves} the kinetic energy exactly at every interaction, and \rev{leaves} the orientations and hence the interaction energy unchanged. In the absence of Vlasov forces the St\"ormer--Verlet integrator \cref{eq:verlet} is symplectic, so the quadratic kinetic energy error remains bounded and oscillatory over very long time intervals, and in particular it does not grow with the number of steps \cite{hairer2006geometric}.

\rev{For a quadratic energy, symplectic Runge--Kutta methods, such as the Gauss collocation methods, are exactly conservative, since they preserve every quadratic invariant of the flow \cite[Section~IV.2]{hairer2006geometric}. The St\"ormer--Verlet scheme is not a Runge--Kutta method, but rather a partitioned Runge--Kutta method, the two-stage Lobatto IIIA--IIIB pair \cite[Section~II.2]{hairer2006geometric}, and} \revii{symplectic partitioned Runge--Kutta methods do not preserve the kinetic energy $\tfrac{1}{2}\vec{\varsigma}\cdot\mat{B}^{-1}\vec{\varsigma}$.} Note that in our scheme the force itself is reconstructed from the particles at every step (\cref{sec:meanfieldforce}), so exact conservation cannot be expected even from a symplectic Runge--Kutta integrator. We opted for a second-order time discretization of \cref{eq:ODEs_2}, because a first-order one perturbs the energy at each step with a definite sign, producing a drift that grows linearly in time. Over the long horizons needed to reach the ordered steady state, this drift would overwhelm the physical effect we seek to capture, namely the rise in temperature that accompanies the emergence of orientational order, as interaction energy is converted into kinetic energy at constant $\mathcal{E}$.

To illustrate the conservation properties of different transport integrators, we perform a brief numerical experiment (\cref{fig:integrators}). Using the transport substep of \rev{the hard-needle simulation from Test~5 \revii{(described in detail in \cref{sec:hardneedle})}}, we compare four schemes: the explicit and symplectic Euler methods (first order), the symmetric St\"ormer--Verlet scheme \cref{eq:verlet} (second order), and Ruth's symplectic integrator \cite{Ruth1983} (third order). Over a horizon of $t=500$, both first-order schemes drift by about a factor of two in total energy, spuriously heating the system. In contrast, the St\"ormer--Verlet scheme maintains the energy plateau within a few percent. The third-order integrator offers only marginal improvement over St\"ormer--Verlet, reproducing similar physics. This confirms that second-order time-stepping accuracy is sufficient to prevent numerical errors from polluting the physical results. For a broader view of symplectic integration, see \cite{SanzSerna1992,LeimkuhlerReich2004,FengQin2010,BlanesCasas2016}; for integrating Hamiltonian dynamics on manifolds, see \cite[Chapter~IV]{hairer2006geometric} and \cite[Chapters~7 and~8]{LeimkuhlerReich2004}.

\begin{figure}[tb]
	\centering
	\subfloat[{full horizon, $t\in[0,500]$}]{\includegraphics[width=0.8\textwidth]{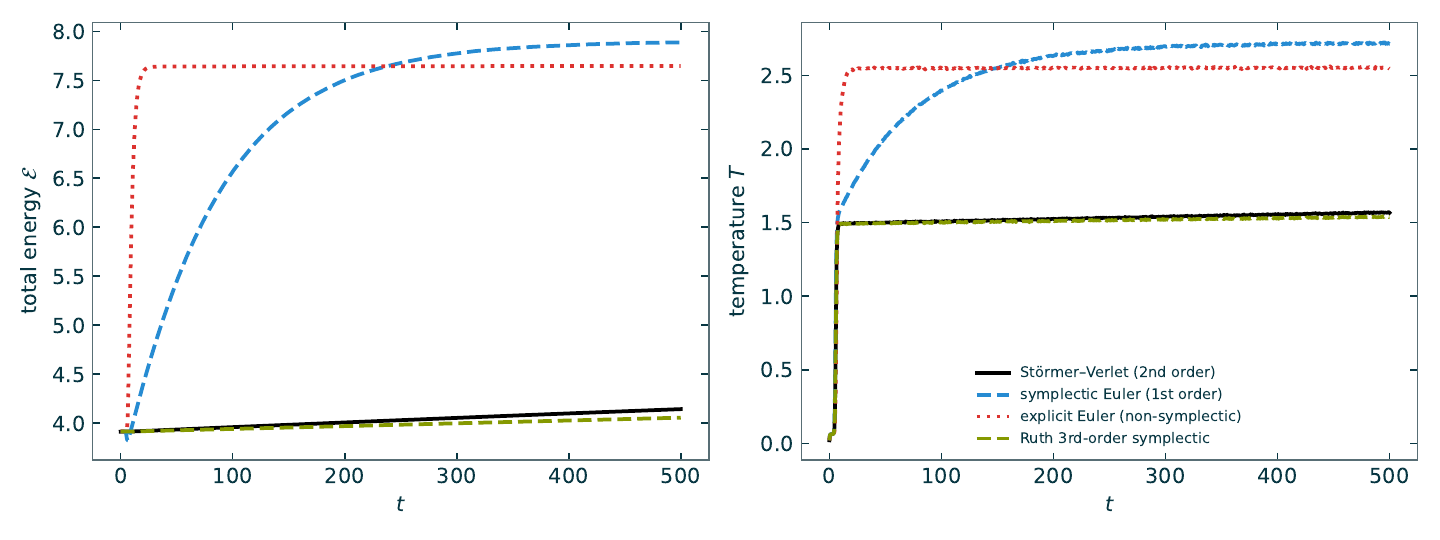}}\\
	\subfloat[{early times, $t\in[0,50]$}]{\includegraphics[width=0.8\textwidth]{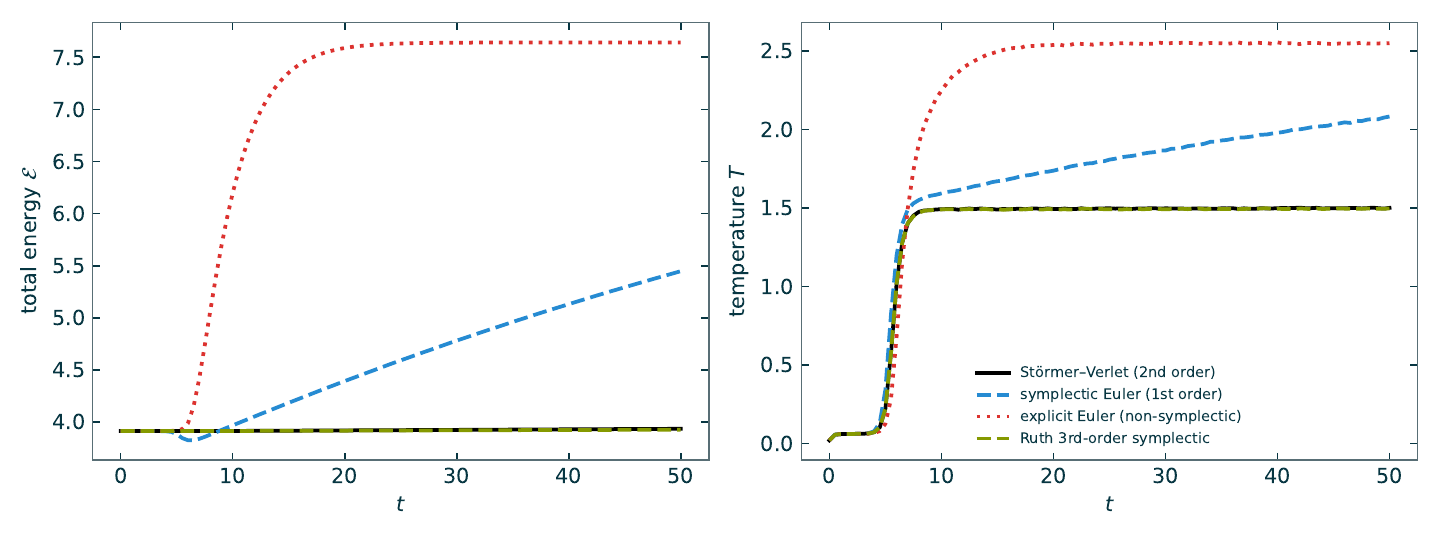}}
	\caption{Total energy and temperature histories of a DSMC simulation of \eqref{eq:Boltzmann2D_calamitic} with hard-needle collision kernel and Onsager potential. The transport substep is integrated by the explicit Euler, symplectic Euler, St\"ormer--Verlet and third-order Ruth schemes. The first-order schemes drift and spuriously heat the system, while the second- and third-order schemes hold the energy plateau and produce the same physics.}
	\label{fig:integrators}
\end{figure}

\subsection{Reconstruction of the mean-field force}
\label{sec:meanfieldforce}
\rev{In the classical DSMC method for \reviv{neutral} monatomic gases the transport step is a free streaming, and it requires no information beyond the velocity of the particle itself.} \reviv{For charged particles the transport is instead driven by the self-consistent field, computed from the particles by the particle-in-cell methods of plasma physics \cite{hockneyEastwood,birdsallLangdon}.}
\rev{In our scheme the transport step is \reviv{likewise} driven by a mean-field force that depends on the unknown distribution.} 
Consequently, the force must be \emph{reconstructed} from the particle ensemble at every time step. This reconstruction is a central piece of our discretization. \revii{If the force only depends on $f$ through a few sample moments, we simply accumulate those moments.} \rev{If the force involves the convolution \cref{eq:conservative_force_generic}, we must reconstruct the full \reviv{particle density} $\rho$ on the mesh of $\mathcal{M}$ defined in \cref{sec:particles}. Borrowing from the particle-in-cell literature \cite{hockneyEastwood,birdsallLangdon}, we mesh the manifold, reconstruct $\rho$ using shape functions that form a partition of unity, evaluate the convolution $\Phi=\mathcal{W}\ast\rho$ by quadrature on the mesh, and finally differentiate the grid potential to recover the force.}

\paragraph{Two-dimensional calamitic molecules}
For the phenomenological quadratic and Kuramoto potentials of \cref{sec:potentials} the force depends on $f$ only through the mean orientation of the sample. In this case the reconstruction amounts to computing the circular mean of the particle orientations,
\begin{equation}\label{eq:circmean}
	\hat{\theta} = \arctan\big(\hat{\nu}^{(y)},\hat{\nu}^{(x)}\big), \qquad \hat{\vec{\nu}}=\sum_{i=1}^N \vec{\nu}_i, \qquad \vec{\nu}_i = (\cos\theta_i,\sin\theta_i),
\end{equation}
at a cost of $\mathcal{O}(N)$ per time step.

For a general interaction potential the force is given by \cref{eq:conservative_force} and the full orientational density $\rho_\theta$ must be reconstructed. We use a cloud-in-cell (CIC) estimator on the grid $\{I_l\}$ of \cref{sec:particles}: each particle distributes its \reviv{weight $1/N$} linearly between the two cells nearest to its orientation, with periodic wrapping at the ends of the grid. The CIC estimator is the second member of a hierarchy of reconstructions by shape functions of increasing smoothness, whose first member is the raw histogram. We choose it because it yields a continuous, piecewise-linear approximation of $\rho_\theta$ with a lower statistical noise than the histogram, while touching only two cells per particle. Reconstructions of this kind are a classical tool of the particle-in-cell methods of plasma physics \cite{hockneyEastwood,birdsallLangdon}. The convolution $\Phi=\mathcal{W}\ast\rho_\theta$ is then computed as a matrix--vector product on the grid, and the force on a particle in the cell $I_l$ is the difference quotient of the grid potential between the two adjacent cells, $-I_3^{-1}(\Phi_{l+1}-\Phi_l)/\Delta\theta$. The cost of this step is $\mathcal{O}(N_\theta^2)$ per time step, independent of the number $N$ of particles, and dominates the cost of the transport substep.

\subsection{Nanbu--Babovsky collisions}
\label{sec:collision}

The physical binary collisions are solved with the Nanbu--Babovsky DSMC approach \cite{nanbu,babovsky}. Babovsky and Illner \cite{babovskyIllner} established its convergence to the solution of the Boltzmann equation; Pareschi and Toscani \cite{pareschi2013interacting} and Pareschi and Russo \cite{pareschi2001time} give systematic accounts. We first describe the method for the Maxwellian kernel, for which all pairs of molecules are equally likely to collide and the collision operator \rev{reads $\mathcal{C}[f,f]=\int (f'f_*'-ff_*)\,d\Xi_*$, where $\Xi=(\vec{v},\nu,\vec{\varsigma})$ collects the phase-space coordinates of a molecule, the subscript $*$ denotes the collision partner, the primes denote the post-collisional states, and the calamitic instance is the right-hand side of \cref{eq:Boltzmann2D_calamitic}}. In this case we can rewrite the collisional operator to highlight its gain and loss parts, by integrating the loss term,
\begin{equation}
	\frac{\partial{f}}{\partial t}  =\rev{\frac{1}{\tau}\int f' f_*' \,d\Xi_* - \frac{1}{\tau}f,}
\end{equation}
and we discretize the time derivative with a first-order Euler scheme to obtain
\begin{equation}\label{eq:collision_convex}
f^{n+1} = \left(1-\frac{\Delta t}{\tau}\right)f^n + \frac{\Delta t}{\tau} C^{+}[f^n,f^n],\qquad C^{+}[f,f]\coloneqq \rev{\int f' f_*' \,d\Xi_*.}
\end{equation}
Since $f^{n+1}$ is a convex combination of two probability density functions\reviv{, provided $\Delta t\leq\tau$}, it is itself a probability density function, and we can apply the classical DSMC probabilistic interpretation.
With probability $1-\Delta t/\tau$ a particle remains unchanged, while with probability $\Delta t/\tau$ it undergoes a binary interaction.

We stress that this rewriting hinges on the Maxwellian kernel. For a general transition probability the loss term is $f\int \mathcal{W}(\Xi,\Xi_*)\,f_*\,d\Xi_*$, whose coefficient depends on the state of the molecule, and the convex-combination structure of \cref{eq:collision_convex} is lost. For hard top-symmetric molecules, Curtiss and Dahler obtained a gain--loss splitting of the collision operator by exploiting the symmetry of the molecules \cite{curtiss,curtissV}. For a generic bounded kernel the standard remedy is to bound the kernel from above and to introduce fictitious (null) collisions that leave the particles unchanged \cite{koura,bird}. This is the route we follow for the physical hard-needle kernel, and we derive it in \cref{sec:kernels}.

\paragraph{Two-dimensional calamitic molecules}
For two-dimensional calamitic molecules the binary interaction follows the collision rule in \cite[Example~2.33]{carrillo2025kinetic}, i.e.
\begin{alignat}{2}
\label{eq:binary_rule_vij} \vec{v}'_i &= \vec{v}_i - (1+e_v)\frac{J}{m} \vec{n},\qquad && \vec{v}'_j  = \vec{v}_j + (1+e_v)\frac{J}{m} \vec{n}, \\
\label{eq:binary_rule_omegaij} \omega'_i  &= \omega_i - (1+e_\omega)J  I^{-1}_3 (\vec{r}_i\times \vec{n}), \qquad &&\omega'_j  = \omega_j + (1+e_\omega)J  I^{-1}_3 (\vec{r}_j\times \vec{n}),
\end{alignat}
with impulse
\begin{equation}\label{eq:binary_rule_J}
J = -\dfrac{V \cdot \vec{n}}{\frac{2}{m} + \left[ I_3^{-1} (\vec{r}_i\times \vec{n})^2 + I^{-1}_3 (\vec{r}_j\times \vec{n})^2 \right]},
\end{equation}
where $\vec{n}$ is the contact normal, $\vec{r}_i,\vec{r}_j$ the contact arms (see \cref{fig:collision}), $V$ the relative contact velocity, $e_v,e_\omega\in[0,1]$ the translational and rotational coefficients of restitution, and $I_3=mL^2/12$ is the moment of inertia of a segment of length $L$ rotating about its center of mass. The orientations are unchanged by the collision, $\theta'=\theta$ and $\theta_*'=\theta_*$. We take $e_v=e_\omega=1$ throughout, so that the collisions are elastic and conserve linear momentum, angular momentum, and energy exactly. \rev{In \cref{algorithm} the collision parameters of this model are sampled as $\psi\sim\mathcal{U}([0,2\pi])$ and $\ell\sim\mathcal{U}(\revvi{[-L/2,L/2]})$, fixing the contact normal $\vec{n}=[\cos\psi,\sin\psi]$ and the contact arms $\vec{r}_i=\ell\vec{\nu}_i$ and $\vec{r}_j=\revvi{\tfrac{L}{2}}\vec{\nu}_j$, with $\vec{\nu}_k=[\cos\theta_k,\sin\theta_k]$.}

\begin{remark}
\label[remark]{rem:arms}
	Two segments of vanishing girth in general position never overlap along a set of positive length, so the first contact between two approaching molecules occurs when an endpoint of one meets the interior of the other. \revvi{The endpoints of a segment of length $L$ lie at distance $L/2$ from its centre.} We therefore place the contact \rev{on the axis of molecule $j$ at distance \revvi{$L/2$} from its centre}, so that $\vec{r}_j=\revvi{\tfrac{L}{2}}\vec{\nu}_j$, and \rev{on the axis of molecule $i$ at a distance drawn uniformly}, so that $\vec{r}_i=\ell\vec{\nu}_i$ with $\ell\sim\mathcal{U}(\revvi{[-L/2,L/2]})$. The resulting rule is asymmetric in the two molecules, but no molecule is privileged by this choice, since the pairs $(i,j)$ are drawn uniformly at every step, so that each molecule takes either role with equal probability.
\end{remark}

The Nanbu--Babovsky procedure is summarized in \cref{algorithm}. The number $N_c$ of interacting pairs per step is set by the collision frequency through a stochastic rounding $\mathrm{Sround}$, which rounds a positive real $x$ to $\lfloor x\rfloor+1$ with probability $x-\lfloor x\rfloor$ and to $\lfloor x\rfloor$ otherwise, so that $\mathbb{E}[\mathrm{Sround}(x)]=x$.

\begin{figure}[htb]
	\centering
	\begin{minipage}{.9\linewidth}
		\begin{algorithm}[H]
			\footnotesize
			\caption{\small \rev{Nanbu--Babovsky algorithm for a generic order-parameter manifold} } \label{algorithm}
				\begin{algorithmic}
					\STATE{Compute the initial \rev{state $\{\vec{v}^0_i,\,\nu^0_i,\,\vec{\varsigma}^0_i\}_{i=1}^N$} by sampling from the initial distribution \rev{$f^0(\vec{v},\nu,\vec{\varsigma})$};}
					\FOR{$n=1$ to $t_f/\Delta t$, given \rev{$\{\vec{v}^n_i,\,\nu^n_i,\,\vec{\varsigma}^n_i\}_{i=1}^N$}}
						\STATE{set $N_c=\textrm{Sround}(N\Delta t / 2 \tau)$, number of interacting pairs;}
						\STATE{select $N_c$ interacting pairs $(i,j)$ uniformly among all the possible ones;}
						\FOR{every pair \rev{$\Xi^n_i$--$\Xi^n_j$, with $\Xi=(\vec{v},\nu,\vec{\varsigma})$}}
							\STATE{\rev{sample the collision parameters of the molecular model, for the calamitic molecules the contact normal $\vec{n}$ and the contact arms $\vec{r}_i$, $\vec{r}_j$ of \cref{fig:collision};}}
							\STATE{\rev{compute the post-collisional state by the binary collision rule of the model, for the calamitic molecules \cref{eq:binary_rule_vij,eq:binary_rule_omegaij} with the impulse \cref{eq:binary_rule_J};}}
							\STATE{set $\vec{v}^{n+1}_i=\vec{v}'_i$, \rev{$\vec{\varsigma}^{n+1}_i=\vec{\varsigma}'_i$}, $\vec{v}^{n+1}_j=\vec{v}'_j$, and \rev{$\vec{\varsigma}^{n+1}_j=\vec{\varsigma}'_j$};}
						\ENDFOR
						\STATE{set $\vec{v}^{n+1}_i=\vec{v}^n_i$ and \rev{$\vec{\varsigma}^{n+1}_i=\vec{\varsigma}^n_i$} for all the particles that have not been collided;}
						\STATE{set \rev{$\nu^{n+1}_i=\nu^n_i$} for all the particles $i=1,\dots,N$;}
					\ENDFOR
				\end{algorithmic}
		\end{algorithm}
	\end{minipage}
\caption*{}
\end{figure}

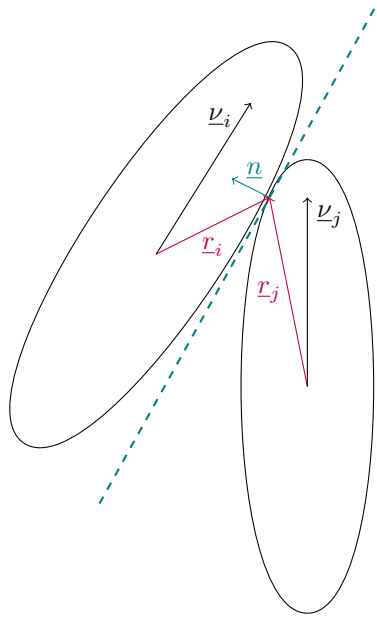
\begin{figure}[htb]
	\centering
	\begin{tikzpicture}
		\begin{pgfonlayer}{nodelayer}
			\node [style=none] (0) at (1, 3.75) {};
			\node [style=none] (1) at (-2.5, -1.5) {};
			\node [style=none] (2) at (1.25, 2.25) {};
			\node [style=none] (3) at (1.25, -3.75) {};
			\node [style=none] (4) at (-0.75, 1) {};
			\node [style=none] (5) at (1.25, -0.75) {};
			\node [style=none] (6) at (1.25, -0.75) {};
			\node [style=none] (7) at (-1.5, -2.3) {};
			\node [style=none] (8) at (2.2, 4.36) {};
			\node [style=none] (9) at (0.75, 1.75) {};
			\node [style=none] (10) at (0.25, 2) {};
			\node [style=none] (11) at (0.55, 2.1) {$\color{verde}\vec{n}$};
			\node [style=none] (12) at (0.1, 2.8) {$\vec{\nu}_i$};
			\node [style=none] (13) at (1.55, 1.5) {$\vec{\nu}_j$};
			\node [style=none] (14) at (0.5, 3) {};
			\node [style=none] (15) at (1.25, 1.75) {};
			\node [style=none] (16) at (0, 1.1) {$\color{rosso}\vec{r}_i$};
			\node [style=none] (17) at (0.75, 0.5) {$\color{rosso}\vec{r}_j$};
		\end{pgfonlayer}
		\begin{pgfonlayer}{edgelayer}
			\draw [bend left=270, looseness=0.50] (0.center) to (1.center);
			\draw [bend left=90, looseness=0.50] (0.center) to (1.center);
			\draw [bend right=90, looseness=0.50] (2.center) to (3.center);
			\draw [bend left=90, looseness=0.50] (2.center) to (3.center);
			\draw [style=GreenDashed] (7.center) to (8.center);
			\draw [style=greenArrow] (9.center) to (10.center);
			\draw [style=ArrowRed] (4.center) to (9.center);
			\draw [style=ArrowRed] (6.center) to (9.center);
			\draw [style=Arrow] (4.center) to (14.center);
			\draw [style=Arrow] (6.center) to (15.center);
		\end{pgfonlayer}
	\end{tikzpicture}
	\caption{Collision of two calamitic molecules in two dimensions. The contact normal $\vec{n}$ and the contact arms $\vec{r}_i$, $\vec{r}_j$ determine the impulse $J$ of \cref{eq:binary_rule_J} and the post-collisional velocities \cref{eq:binary_rule_vij,eq:binary_rule_omegaij}.}
	\label{fig:collision}
\end{figure}

\subsection{Collision kernels and Bird's method}
\label{sec:kernels}
\Cref{algorithm} uses a Maxwellian kernel, where all pairs are equally likely to collide. \rev{For a general molecule, the collision operator \cite{carrillo2025kinetic} takes the form proposed by Cercignani and Lampis \cite{cercignaniLampis}, allowing forward and reverse transition probabilities to differ:
\begin{equation}\label{eq:collision_operator_generic}
	\mathcal{C}[f,f](\Xi) = \int \Big[ W(\Xi',\Xi_*'\mapsto\Xi,\Xi_*)\,f'f_*' - W(\Xi,\Xi_*\mapsto\Xi',\Xi_*')\,ff_* \Big]\, d\Xi_*\,d\Xi'\,d\Xi_*'.
\end{equation}
Integrating the forward transition probability over all post-collisional states yields a pair's total collision rate:
\begin{equation}\label{eq:pair_rate}
	\beta(\Xi,\Xi_*) = \int W(\Xi,\Xi_*\mapsto\Xi',\Xi_*')\,d\Xi'\,d\Xi_*'.
\end{equation}
Because this loss coefficient $\beta$ depends on the molecular state, the convex-combination structure of \cref{eq:collision_convex} breaks down.} 

To restore this structure, we use Bird's acceptance--rejection technique \cite{bird}, which relies on Koura's null-collision principle \cite{koura}. Although Bird's approach predates Nanbu and Babovsky's, its convergence to the Boltzmann equation solution was rigorously established later \cite{Wagner1992}. \rev{By bounding the pair rate from above ($\beta \leq w_{\max}$), adding and subtracting $w_{\max}\,f f_*$ in the loss term, and noting that $f_*$ integrates to one, we obtain:
\begin{equation}\label{eq:bird_split}
	\rev{\frac{\partial f}{\partial t} = \int W(\Xi',\Xi_*'\mapsto\Xi,\Xi_*)\, f' f_*'\, d\Xi_*\,d\Xi'\,d\Xi_*' + \int \big(w_{\max}-\beta\big)\, f f_*\, d\Xi_* - w_{\max} f.}
\end{equation}
The loss term $-w_{\max}f$ now carries a constant coefficient $w_{\max}$ that is independent of the molecule's state. The gain term splits into two non-negative parts: true collisions (weighted by the forward transition probability) and null collisions (weighted by $w_{\max}-\beta$, which leave the pair unchanged). Together, these gain terms integrate exactly to $w_{\max}$.} \revii{Applying a first-order Euler discretization in time thus recovers a convex combination of two probability density functions,
\begin{equation}\label{eq:bird_convex}
	f^{n+1} = \left(1-w_{\max}\Delta t\right) f^n + w_{\max}\Delta t\; C^{+}_{\max}[f^n,f^n],
\end{equation}
where the normalised gain
\begin{equation}\label{eq:bird_gain}
	C^{+}_{\max}[f,f] = \frac{1}{w_{\max}}\int \Big[ \int W(\Xi',\Xi_*'\mapsto\Xi,\Xi_*)\, f' f_*'\, d\Xi'\,d\Xi_*' + \big(w_{\max}-\beta\big)\, f f_* \Big]\, d\Xi_*
\end{equation}
is itself a probability density, and the combination is valid for $\Delta t\leq 1/w_{\max}$. This is \cref{eq:collision_convex} with $1/\tau$ replaced by $w_{\max}$, and it carries the same probabilistic interpretation. With probability $1-w_{\max}\Delta t$ a particle is left unchanged, and with probability $w_{\max}\Delta t$ it takes part in a candidate interaction.} \rev{The right-hand side describes an equivalent dynamics in which every pair collides with the constant frequency $w_{\max}$. Each of these collisions is a true collision with probability $\beta/w_{\max}$, in which case the post-collisional state is drawn from the normalised transition density $W(\Xi,\Xi_*\mapsto\cdot)/\beta(\Xi,\Xi_*)$, and a \emph{null} collision otherwise. The scheme samples only the forward transition probability, so no symmetry between the forward and the reverse transition probabilities is required. The reciprocity principle \cite{carrillo2025kinetic} relates these probabilities and enters the $H$-theorem, not the algorithm.} \reviv{This acceptance--rejection step is where the scheme differs most from the classical DSMC method, since the transition probability need not be symmetric.}

\paragraph{Two-dimensional calamitic molecules}
The physical kernel of two-dimensional calamitic molecules \cite[Example~2.33]{carrillo2025kinetic} is the hard-needle kernel $W(\Xi,\Xi_*)=|\vec{V}\cdot\vec{n}|\,S(\theta,\theta_*)$, with cross-section $S=L\,|\sin(\theta-\theta_*)|$, for which nearly parallel rods have a vanishing cross-section. \rev{Here $W(\Xi,\Xi_*)$ is the density of the forward transition probability over the impact parameters $(\psi,\ell)$ of \cref{algorithm}, the post-collisional state being their deterministic outcome \cref{eq:binary_rule_vij,eq:binary_rule_omegaij}, so the pair rate \cref{eq:pair_rate} is the average of $W$ over the impact parameters and a pointwise bound $W\leq w_{\max}$ bounds $\beta$.} In practice one draws $\mathrm{Sround}(w_{\max} N \Delta t/2)$ candidate pairs, \rev{samples the impact parameters uniformly,} evaluates the kernel $w=|\vec{V}\cdot\vec{n}|\,L|\sin(\theta-\theta_*)|$ on each candidate, and accepts the collision with probability $w/w_{\max}$. \rev{The acceptance realises simultaneously the choice between a true and a null collision, with probability $\beta/w_{\max}$, and the sampling of the impact parameters with density proportional to $w$, a standard property of acceptance--rejection sampling.} The bound $w_{\max}$ is adjusted adaptively during the simulation.
The procedure is summarized in \cref{algorithm:bird}, and parallels \cref{algorithm} with the acceptance--rejection step added. 

\begin{figure}[htb]
	\centering
	\begin{minipage}{.9\linewidth}
		\begin{algorithm}[H]
			\footnotesize
			\caption{\small \rev{Bird's method for a generic bounded collision kernel} } \label{algorithm:bird}
				\begin{algorithmic}
					\STATE{Compute the initial \rev{state $\{\vec{v}^0_i,\,\nu^0_i,\,\vec{\varsigma}^0_i\}_{i=1}^N$} by sampling from the initial distribution \rev{$f^0(\vec{v},\nu,\vec{\varsigma})$}, and initialize the bound $w_{\max}$;}
					\FOR{$n=1$ to $t_f/\Delta t$, given \rev{$\{\vec{v}^n_i,\,\nu^n_i,\,\vec{\varsigma}^n_i\}_{i=1}^N$}}
						\STATE{set $N_c=\mathrm{Sround}(w_{\max} N \Delta t / 2)$, number of candidate pairs;}
						\STATE{select $N_c$ candidate pairs $(i,j)$ uniformly among all the possible ones;}
						\FOR{every candidate pair \rev{$\Xi^n_i$--$\Xi^n_j$}}
							\STATE{\rev{sample the collision parameters of the model, as in \cref{algorithm};}}
							\STATE{evaluate the kernel \rev{$w$ on the sampled parameters, for the hard needles $w=|\vec{V}\cdot\vec{n}|\,L\,|\sin(\theta_i-\theta_j)|$,} and draw $u\sim\mathcal{U}([0,1])$;}
							\IF{$u<w/w_{\max}$}
								\STATE{\rev{perform the collision according to the rule of the model, for the calamitic molecules \cref{eq:binary_rule_vij,eq:binary_rule_omegaij} with the impulse \cref{eq:binary_rule_J};}}
							\ELSE
								\STATE{leave the pair unchanged (null collision);}
							\ENDIF
						\ENDFOR
						\STATE{leave all the non-selected particles unchanged, and set \rev{$\nu^{n+1}_i=\nu^n_i$} for all $i=1,\dots,N$;}
						\STATE{update the bound $w_{\max}$ adaptively from the kernel values sampled during the step;}
					\ENDFOR
				\end{algorithmic}
		\end{algorithm}
	\end{minipage}
\caption*{}
\end{figure}

\subsection{A handed exchange rule}
\label{sec:handedrule}
\rev{We now construct a collision rule for calamitic molecules whose forward and reverse transition probabilities genuinely differ. We test it numerically in \cref{sec:handed}. Given a selected pair, sample a direction $\vec{n}=[\cos\psi,\sin\psi]$ with $\psi\sim\mathcal{U}([0,2\pi])$\revvi{, replaced by $-\vec{n}$ if $\vec{n}\cdot(\vec{\nu}_i+\vec{\nu}_j)<0$\footnote{We neglect the case $\vec{\nu}_i=-\vec{\nu}_j$, which corresponds to a configuration with two anti-parallel nematic molecules. This is not considered in our collision model since it is not possible to define a unique collision point. Since such configurations occur with probability zero, we expect neglecting this configuration to have no qualitative implications in our simulations.} so that the direction lies in the half circle fixed by the pair,} and a \emph{one-sided} angle $\chi\sim\mathcal{U}([0,\pi])$, and rotate
\begin{equation}\label{eq:handed_rule}
	\begin{pmatrix} a \\ b \end{pmatrix} =
	\begin{pmatrix} \sqrt{m/4}\;\, \vec{g}\cdot\vec{n} \\ \reviii{\sqrt{I_3/4}\;\,(\omega_i-\omega_j)} \end{pmatrix}
	\;\mapsto\;
	\begin{pmatrix} \cos\chi & -\sin\chi \\ \sin\chi & \phantom{-}\cos\chi \end{pmatrix}
	\begin{pmatrix} a \\ b \end{pmatrix},
\end{equation}
where $\vec{g}=\vec{v}_i-\vec{v}_j$ is the relative velocity.} \reviii{Denoting by $\Delta$ the increment produced by the rotation, the pair is updated by
\begin{equation}\label{eq:handed_update}
	\vec{v}_i' = \vec{v}_i + \tfrac{1}{2}\Delta(\vec{g}\cdot\vec{n})\,\vec{n}, \quad
	\vec{v}_j' = \vec{v}_j - \tfrac{1}{2}\Delta(\vec{g}\cdot\vec{n})\,\vec{n}, \quad
	\omega_i' = \omega_i + \tfrac{1}{2}\Delta\omega, \quad
	\omega_j' = \omega_j - \tfrac{1}{2}\Delta\omega,
\end{equation}
where $\Delta\omega=\Delta(\omega_i-\omega_j)$, the tangential component of $\vec{g}$ and the orientations being left unchanged. The quantity
\begin{equation}\label{eq:handed_energy}
	a^2+b^2 = \frac{m}{4}\,(\vec{g}\cdot\vec{n})^2 + \frac{I_3}{4}\,(\omega_i-\omega_j)^2
\end{equation}
is the part of the kinetic energy of the pair carried by the normal component of the relative velocity and by the relative spin. The rotation preserves it and leaves the mean velocity of the pair, its tangential relative velocity, and its mean spin untouched, so every such map conserves the linear momentum, the total spin $I_3(\omega_i+\omega_j)$, and the kinetic energy of the pair exactly. The generalised angular momentum defined in \cite[Theorem~3.11]{carrillo2025kinetic} is conserved as well.} \rev{Each map is a measure-preserving bijection of the pair phase space, applied at the constant Maxwellian rate $1/\tau$, so both sides of the reciprocity principle \cite{carrillo2025kinetic} equal $1/\tau$ and the $H$-theorem applies. \revvi{The restriction on $\vec{n}$ involves the order parameters alone, which the collision leaves unchanged, so the same half circle is available before and after the collision.} However, detailed balance fails. Undoing a collision requires the rotation by $-\chi$. \revvi{Keeping the angle and reversing the direction does exactly that, and it is what the restriction forbids.}}

\subsection{The BGK operator}
\label{sec:bgk}
\reviv{In the strongly collisional regime we can replace the binary collisions with the BGK approximation} of the collision operator \cite{bgk}, which relaxes the distribution towards the local Maxwellian \reviv{$\bar{f}_{\vec{U},\Omega,T}$} carrying the same collision invariants as $f$\reviv{, with $\vec{U}$, $\Omega$ and $T$ the mean velocity, the mean conjugate momentum and the temperature of $f$},
\begin{equation}\label{eq:bgk_operator}
	\mathcal{Q}_{\mathrm{BGK}}[f] = \frac{1}{\tau_{\mathrm{r}}}\big(\reviv{\bar{f}_{\vec{U},\Omega,T}}-f\big),
\end{equation}
with $1/\tau_{\mathrm{r}}>0$ the relaxation frequency. We substitute \cref{eq:bgk_operator} into the collisional substep $\partial_t f = \mathcal{Q}_{\mathrm{BGK}}[f]$, freeze the Maxwellian over the time step, and discretize the time derivative with a first-order Euler scheme. This gives
\begin{equation}\label{eq:bgk_convex}
	f^{n+1} = \left(1-\frac{\Delta t}{\tau_{\mathrm{r}}}\right)f^n + \frac{\Delta t}{\tau_{\mathrm{r}}}\,\reviv{\bar{f}_{\vec{U}^n,\Omega^n,T^n}}.
\end{equation}
Since $f^{n+1}$ is a convex combination of two probability density functions, provided $\Delta t\leq\tau_{\mathrm{r}}$, it is itself a probability density function, and it carries the DSMC probabilistic interpretation \reviv{of \cref{eq:collision_convex}}. With probability $1-\Delta t/\tau_{\mathrm{r}}$ a particle is left unchanged, and with probability $\Delta t/\tau_{\mathrm{r}}$ its \rev{velocity and conjugate momentum are resampled from the corresponding Maxwellian $\reviv{\bar{f}_{\vec{U}^n,\Omega^n,T^n}}$, the order parameters} being left unchanged. The procedure is summarized in \cref{algorithm:bgk}.

The collisions in \cref{sec:collision} are instantaneous momentum exchanges at the point of contact that leave orientations unchanged. This has two key consequences. 

First, any function of the orientation alone acts as a collision invariant, in addition to the four invariants identified in \cite{carrillo2025kinetic}. As a result, the collision operator pointwise conserves the entire orientational density. Retaining the orientations in \cref{eq:bgk_convex} essentially relaxes $f$ toward the product of its \rev{order-parameter} marginal and the velocity Gaussians\rev{, which in the calamitic chart of \cref{sec:calamitic} reads}:
\begin{equation}\label{eq:bgk_target}
	\reviv{\bar{f}_{\vec{U},\Omega,T}}(\vec{v},\theta,\omega) = \rho_\theta(\theta)\,\frac{m}{2\pi T}\exp\left(-\frac{m|\vec{v}-\vec{U}|^2}{2T}\right)\sqrt{\frac{I_3}{2\pi T}}\exp\left(-\frac{I_3(\omega-\Omega)^2}{2T}\right).
\end{equation}

Second, these Gaussians must capture \emph{every} moment conserved by the collisions. In particular, failing to include the mean angular velocity $\Omega$ would spuriously damp the sample's mean spin at rate $1/\tau_{\mathrm{r}}$. This formulation is a classic maximum-entropy closure: $\bar{f}$ maximizes entropy among all distributions that preserve the required moments. In this light, our ansatz is the simplest member of established kinetic-model hierarchies and moment closures \cite{grossJackson1959,levermore,levermoreMorokoff}.

The main drawback of using a single relaxation time is that all moments of $f$ decay at identical rates. Consequently, the Chapman--Enskog transport coefficients are artificially coupled, leading to a Prandtl number of one rather than the $2/3$ expected for a monatomic gas \cite[Section~3.1]{cercignani}. %

\begin{figure}[htb]
	\centering
	\begin{minipage}{.9\linewidth}
		\begin{algorithm}[H]
			\footnotesize
			\caption{\small{BGK collision step for the order-parameter dynamics} } \label{algorithm:bgk}
				\begin{algorithmic}
					\STATE{Compute the initial \rev{state $\{\vec{v}^0_i,\,\nu^0_i,\,\vec{\varsigma}^0_i\}_{i=1}^N$} by sampling from the initial distribution \rev{$f^0(\vec{v},\nu,\vec{\varsigma})$};}
					\FOR{$n=1$ to $t_f/\Delta t$, given \rev{$\{\vec{v}^n_i,\,\nu^n_i,\,\vec{\varsigma}^n_i\}_{i=1}^N$}}
						\STATE{compute the mass, the mean velocity, the \rev{mean conjugate momentum}, and the temperature of the sample and form the Maxwellian $\reviv{\bar{f}_{\vec{U}^n,\Omega^n,T^n}}$ of \eqref{eq:bgk_target};}
						\FOR{every particle $i=1,\dots,N$}
							\STATE{sample a uniform random number $u_i\sim\mathcal{U}([0,1])$;}
							\IF{$u_i<\Delta t/\tau_{\mathrm{r}}$}
								\STATE{\rev{resample $(\vec{v}^{n+1}_i,\,\vec{\varsigma}^{n+1}_i)$ from the Maxwellian $\reviv{\bar{f}_{\vec{U}^n,\Omega^n,T^n}}$ in the velocity and conjugate-momentum variables;}}
							\ELSE
								\STATE{set $\vec{v}^{n+1}_i=\vec{v}^n_i$ and \rev{$\vec{\varsigma}^{n+1}_i=\vec{\varsigma}^n_i$};}
							\ENDIF
							\STATE{set \rev{$\nu^{n+1}_i=\nu^n_i$};}
						\ENDFOR
					\ENDFOR
				\end{algorithmic}
		\end{algorithm}
	\end{minipage}
\caption*{}
\end{figure}

The same BGK collision step is the mechanism behind the thermostat of \cref{sec:thermostat}. There the Maxwellian is held fixed at the bath temperature $T_{\mathrm{bath}}$ and the frequency is $1/\tau_{\mathrm{bath}}$, so that the resampling drives the system towards the canonical equilibrium instead of conserving the energy.

\subsection{Isothermal simulations: a BGK/Andersen thermostat}
\label{sec:thermostat}
Because\linebreak both the collision operator and transport conserve total energy at the continuous level \cref{eq:totalenergy}, the scheme naturally simulates the microcanonical (constant-energy) ensemble. To simulate the canonical (isothermal) ensemble and trace the model's phase diagram as the temperature $T_{\mathrm{bath}}$ is varied, we introduce an Andersen thermostat \cite{andersen}. In this approach, each particle independently undergoes a ``bath collision'' with frequency $1/\tau_{\mathrm{bath}}$. During such a collision, the particle's velocities are resampled from a Maxwellian at $T_{\mathrm{bath}}$:
\begin{equation}\label{eq:thermostat}
	\vec{v}\sim\mathcal{N}\!\big(0,\sqrt{T_{\mathrm{bath}}/m}\,\big), \qquad \omega\sim\mathcal{N}\!\big(0,\sqrt{T_{\mathrm{bath}}/I_3}\,\big),
\end{equation}
while its orientation remains unchanged. This resampling occurs with probability $\Delta t/\tau_{\mathrm{bath}}$ per time step. Conceptually, this is the particle realization of a \reviv{BGK} relaxation \cite{bgk}, $\mathcal{Q}_{\mathrm{bath}}[f]=\tau_{\mathrm{bath}}^{-1}(\bar{f}_{T_{\mathrm{bath}}}-f)$, driving the system toward the canonical equilibrium at $T_{\mathrm{bath}}$. We apply this thermostat once per time step, immediately following the collision substep in \cref{eq:fullstep}.

\section{Mean-field potentials}
\label{sec:potentials}

The mean-field force $\mathcal{V}$ \revjac{of \cref{eq:boltzmann_generic}, that is of \cref{eq:Boltzmann2D_calamitic} for the two dimensional calamitic molecules,} is determined by the interaction potential $\mathcal{W}$\revii{, through the convolution \cref{eq:conservative_force_generic}. On a generic order-parameter manifold $\mathcal{W}$ is a function on $\mathcal{M}\times\mathcal{M}$. \revjac{On the circle it has to be $2\pi$-periodic in each argument, as observed in \cref{sec:transport}.} If $\mathcal{W}$ is symmetric, $\mathcal{W}(\nu,\nu_*)=\mathcal{W}(\nu_*,\nu)$, then the total energy is conserved by \cref{prop:energy}. If $\mathcal{W}$ is invariant under the group action, $\mathcal{W}(\mathcal{A}_{\reviii{\mat{Q}}}\nu,\mathcal{A}_{\reviii{\mat{Q}}}\nu_*)=\mathcal{W}(\nu,\nu_*)$ for every rotation $\reviii{\mat{Q}}$, then no order parameter is privileged and any order generated by the group action $\mathcal{A}$ starting from a solution $\nu_*$ can be achieved.}

\revii{When the kernel takes the form of a finite sum of separable terms,
\begin{equation}\label{eq:separable_kernel}
	\mathcal{W}(\nu,\nu_*) = \sum_{k=1}^{K} c_k\,\varphi_k(\nu)\,\varphi_k(\nu_*),
\end{equation}
the mean-field potential $\Phi=\mathcal{W}\ast\rho$ depends entirely on the $K$ moments $\int_{\mathcal{M}}\varphi_k\,\rho\,d\nu$ of \reviv{the density} $\rho$. In this case, we can reconstruct the force simply by accumulating those sample moments, which costs only $\mathcal{O}(N)$ per time step. For generic, non-separable kernels, however, we must reconstruct the full density on a mesh of $\mathcal{M}$ and evaluate the convolution via numerical quadrature.}

We consider three choices of potential for the running example of two-dimensional calamitic molecules, of increasing physical fidelity. Each favor alignment of the molecules. The first two are phenomenological forces frequently used to model synchronization and alignment. The third is the excluded-volume potential of Onsager's theory of liquid crystals \cite{onsager}.
\subsection{Quadratic potential}
\label{sec:quadratic}
The harmonic mean-field potential \cite[Example~3.21]{carrillo2025kinetic}
\begin{equation}\label{eq:quadratic_potential}
	\mathcal{W}(\theta,\omega) = \tfrac12\alpha\big|\theta-\hat{\theta}\big|^2 + \beta\,\omega\theta
\end{equation}
yields the linear restoring force
\begin{equation}\label{eq:quadratic_force}
	\mathcal{V}(\theta,\omega) = -\alpha\left(\theta-\hat{\theta}\right) - \beta\,\omega,
\end{equation}
with $\alpha>0$ the interaction strength and $\beta\geq 0$ a friction. The characteristics \cref{eq:ODEs_2} are then linear, with Jacobian eigenvalues $\mu_{1,2}=\tfrac12\big(-\beta\pm\sqrt{\beta^2-4\alpha}\big)$, so that the fixed point $(\hat{\theta},0)$ is a center for $\beta=0$ and $\alpha>0$, a saddle for $\alpha<0$, a stable spiral for $\beta>0$ and $\alpha>\beta^2/4$, and a stable node for $\beta>0$ and $\alpha\in(0,\beta^2/4)$. In the tests of \cref{sec:alignment} we take the frictionless case $\beta=0$ and $\alpha=1$, the damping being supplied by the collisions rather than by the force. Unlike the two potentials below, \cref{eq:quadratic_potential} is posed directly at the mean-field level. It is not generated by a symmetric pair kernel on the circle, since the term $|\theta-\hat{\theta}|^2$ depends on the chosen chart. \revjac{Nor is it periodic. It makes sense on the circle only after a periodic extension, whose force jumps at the cut of the chart, and we evaluate the force \cref{eq:quadratic_force} in the chart, so molecules on the far side of the cut are pushed toward the mean orientation the long way round.} As a result, the energy conservation of \cref{prop:energy} does not apply, and we do not track interaction energy during these tests. \revii{Because the force only depends on the sample through $\hat{\theta}$, its reconstruction uses the cheaper method from \cref{sec:meanfieldforce}.}

\subsection{Kuramoto potential}
\label{sec:kuramoto}
The Kuramoto model of coupled oscillators defines the interaction energy between a pair of molecules via the attractive \emph{pair kernel}:
\begin{equation}\label{eq:cosine_potential}
	\mathcal{W}(\theta,\theta_*) = -\cos(\theta-\theta_*).
\end{equation}
Averaging this kernel against the empirical density yields the \emph{mean-field potential}. By the cosine addition formula, this simplifies to $-R\cos(\theta-\hat{\theta})$, where $R\,e^{i\hat{\theta}} = N^{-1}\sum_i e^{i\theta_i}$ is the sample's resultant (\cref{eq:circmean}). Differentiating this potential gives the \emph{Kuramoto force}, $-R\sin(\theta-\hat{\theta})$. Throughout this paper, we reserve the term ``Kuramoto potential'' strictly for the pair kernel \cref{eq:cosine_potential}. In our numerical tests, however, we use a frozen-amplitude version of the force where the resultant length $R$ is replaced by a constant interaction strength of one:
\begin{equation}\label{eq:cosine_force}
	\mathcal{V}(\theta) = -\sin\left(\theta-\hat{\theta}\right).
\end{equation}
Because this force depends solely on the sample's mean orientation, we reconstruct it using only the circular mean \cref{eq:circmean}. \revii{This computational simplification works precisely because the pair kernel \cref{eq:cosine_potential} is separable ($-\cos(\theta-\theta_*) = -\cos\theta\cos\theta_* - \sin\theta\sin\theta_*$), providing an instance of \cref{eq:separable_kernel} with $K=2$.}

\subsection{Onsager potential}
\label{sec:onsagerpotential}
The Onsager excluded-volume potential of two-dimensional needles \cite{onsager,virgaOnsager},
\begin{equation}\label{eq:onsager}
	\mathcal{W}(\theta,\theta_*) = L^2\,|\sin(\theta-\theta_*)|,
\end{equation}
assigns minimum interaction energy to parallel and anti-parallel rods, for which the excluded volume vanishes, and thus favors alignment along a common, head--tail symmetric director.
\revii{This kernel is not a finite sum of separable terms, its Fourier expansion carrying every even harmonic, so it is the case in which the whole orientational density has to be reconstructed on the grid.}

At temperature $T$, the orientational marginal of the canonical equilibrium of the mean-field dynamics satisfies the Onsager self-consistency equation \cite{onsager,fatkullinSlastikov}:
\begin{equation}\label{eq:selfconsistency}
	\rho(\theta) = \frac{1}{Z}\exp\left(-\frac{(\mathcal{W}\ast\rho)(\theta)}{T}\right), \qquad Z = \int_{-\pi}^{\pi}\exp\left(-\frac{(\mathcal{W}\ast\rho)(\theta)}{T}\right)d\theta.
\end{equation}
Here, the unit-mass density $\rho$ represents $\rho_\theta$. \revjac{\Cref{eq:selfconsistency} holds only at equilibrium. The steady states of \cref{eq:Boltzmann2D_calamitic} with the thermostat of \cref{sec:thermostat} have Maxwellian marginals in $\vec{v}$ and $\omega$ at the temperature $T_{\mathrm{bath}}$, and an orientational marginal $\rho_\theta$ that solves \cref{eq:selfconsistency} \cite[Section~4(b)]{farrellMalekSoucekZerbinatiLC2026}.} \Cref{eq:selfconsistency} is simply the standard Gibbs ansatz from classical mean-field liquid crystal theory \cite{onsager,virgaOnsager}, formulated variationally as in \cite{farrellMalekSoucekZerbinatiLC2026}. Because collisions thermalize the velocities without altering the orientations, the final orientational equilibrium stems strictly from a thermodynamic competition between mean-field energy (driving alignment) and orientational entropy (driving disorder). While the completely isotropic state $\rho_0=1/(2\pi)$ trivially solves \cref{eq:selfconsistency} at all temperatures, it is not always stable. The following proposition performs a linear stability analysis on \cref{eq:selfconsistency}---paralleling Fatkullin \& Slastikov's study of the Onsager functional \cite{fatkullinSlastikov}---to identify the exact temperature where the isotropic state loses stability.

\begin{proposition}
\label{prop:critical}
	\reviv{Let $\mathcal{W}$ be the Onsager potential \cref{eq:onsager}. The isotropic state $\rho_0 = 1/(2\pi)$ is a linearly stable solution of \cref{eq:selfconsistency} if and only if
	\begin{equation}\label{eq:Tc}
		T > T_c = \frac{2L^2}{3\pi},
	\end{equation}
	and the mode that destabilizes it at $T=T_c$ is the nematic, head--tail symmetric mode $\cos(2(\theta-\theta_0))$.}
\end{proposition}
\begin{proof}
	The kernel is even and $\pi$-periodic, so its Fourier expansion contains only even cosine harmonics,
	\begin{equation*}
		L^2|\sin\varphi| = L^2\left(\frac{2}{\pi} - \frac{4}{\pi}\sum_{k\geq 1}\frac{\cos(2k\varphi)}{4k^2-1}\right),
	\end{equation*}
	and, \revjac{since the kernel depends only on $\theta-\theta_*$,} the convolution acts diagonally on the Fourier modes \revjac{for every $\theta_0$}, $\mathcal{W}\ast\cos(2k(\,\cdot\,\revjac{-\theta_0})) = w_{2k}\cos(2k(\theta\revjac{-\theta_0}))$ with
	\begin{equation*}
		w_{2k} = \int_{-\pi}^{\pi} L^2|\sin\varphi|\cos(2k\varphi)\,d\varphi = -\frac{4L^2}{4k^2-1}\revjac{, \qquad k\geq 0}.
	\end{equation*}
	\revjac{In particular $w_0 = 4L^2$.} The odd Fourier coefficients vanish, $w_{2k+1}=0$, because $|\sin\varphi|$ has period $\pi$. A perturbation along an odd harmonic therefore produces no change in the mean-field potential $\mathcal{W}\ast g$ at first order, generates no force, and neither grows nor decays. In this sense polar perturbations are stable. \revjac{We now linearize \cref{eq:selfconsistency} around the unique constant steady state with unitary mass $\rho_0=1/(2\pi)$. Let $\rho = \rho_0 + \varepsilon\, g$ with $\int_{-\pi}^{\pi} g\,d\theta=0$. The convolution is linear, so $\mathcal{W}\ast\rho = c + \varepsilon\,\mathcal{W}\ast g$ with the constant $c=\rho_0 w_0 = 2L^2/\pi$, and $\int_{-\pi}^{\pi}\mathcal{W}\ast g\,d\theta = w_0\int_{-\pi}^{\pi} g\,d\theta = 0$. Expanding the exponential to first order in $\varepsilon$ gives}
	\begin{align*}
		\revjac{e^{-(\mathcal{W}\ast\rho)/T}} &\revjac{= e^{-c/T}\Big(1 - \frac{\varepsilon}{T}\,\mathcal{W}\ast g\Big) + O(\varepsilon^2),} \\
		\revjac{Z} &\revjac{= 2\pi\, e^{-c/T} + O(\varepsilon^2),} \\
		\revjac{\frac{e^{-(\mathcal{W}\ast\rho)/T}}{Z}} &\revjac{= \rho_0 - \varepsilon\,\frac{\rho_0}{T}\,\mathcal{W}\ast g + O(\varepsilon^2).}
	\end{align*}
	\revjac{The factor $e^{-c/T}$ cancels in the quotient, so the constant part of the potential plays no role. Matching the first-order terms with $\rho=\rho_0+\varepsilon g$ yields the linearized equation}
	\begin{equation*}
		g = -\frac{\rho_0}{T}\,\mathcal{W}\ast g.
	\end{equation*}
	\revjac{The linearized map $g\mapsto -(\rho_0/T)\,\mathcal{W}\ast g$ is diagonal in this basis. It annihilates the odd harmonics and multiplies $\cos(2k(\theta-\theta_0))$ by}
	\begin{equation*}
		\revjac{\lambda_{2k} = -\frac{\rho_0\, w_{2k}}{T} = \frac{2L^2}{\pi(4k^2-1)\,T}, \qquad k\geq 1,}
	\end{equation*}
	\revjac{for every $\theta_0$. The isotropic state is linearly stable when every multiplier is below one, which is also the condition for the second variation of the Onsager free energy \cite{fatkullinSlastikov} at $\rho_0$ to be positive. Since $\lambda_{2k}$ decreases with $k$, this reduces to $\lambda_2<1$, that is $T>2L^2/(3\pi)=T_c$. At $T=T_c$ the multiplier of $\cos(2(\theta-\theta_0))$ equals one for every $\theta_0$, so the nematic mode becomes unstable, with the director $\theta_0$ left undetermined by the rotational symmetry.}
\end{proof}

\section{Numerical tests \texorpdfstring{\revii{for two-dimensional calamitic molecules}}{for two-dimensional calamitic molecules}}
\label{sec:experiments}
We validate the scheme on a sequence of tests with two-dimensional calamitic molecules\revjac{, that is on \cref{eq:Boltzmann2D_calamitic}}. Unless otherwise stated we use $N=10^7$ particles, molecular mass $m=1$, \rev{moment of inertia $I_3=1$}, elastic collisions ($e_v=e_\omega=1$), time step $\Delta t=0.05$, and final time $t_f=50$, corresponding to $10^3$ time steps. \rev{In the tests driven by the Onsager potential (\cref{sec:onsager,sec:transition,sec:hardneedle}) the rod length is $L=\sqrt{12}$, consistently with $I_3=mL^2/12=1$, and it sets the strength of the potential and of the hard-needle cross-section. In the tests of \cref{sec:thermalization,sec:alignment} the length enters the dynamics only through the contact arms of the collision rule \revii{of \cref{rem:arms}}, and there we take unit length.} The orientational grid uses $N_\theta=256$ cells. All the tests use the Maxwellian kernel, except the hard-needle test of \cref{sec:hardneedle}, which uses the hard-needle kernel of \cref{sec:kernels}\rev{, and the test of \cref{sec:handed}, which uses the handed exchange rule of \cref{sec:handedrule}}. The rescaled collision frequency is $1/\tau=10$ in the thermalization tests and $1/\tau=4$ in the mean-field tests. \revjac{Videos of the time evolution of the $(v_x,v_y)$ and $(\theta,\omega)$ densities of every test are available at \cite{carrilloFarrellMedagliaZerbinatiVideos2026}.}

Measuring the emergence of order on a manifold requires some care. We borrow the relevant notions from directional statistics \cite{mardiaJupp}. The fundamental difficulty is already visible on the circle: two molecules pointing at $1^\circ$ and $359^\circ$ are almost aligned, yet their arithmetic mean is $180^\circ$. Furthermore, this mean changes if we redefine the branch cut of the angular chart. Averages of chart coordinates are therefore meaningless. Instead, we must average the \emph{embedded} sample, which means averaging the unit vectors themselves.

For circular data, the natural moments are those of the random vector $(\cos\theta,\sin\theta)$. The mean resultant length, $R_1 = |N^{-1}\sum_i e^{i\theta_i}|$, measures polar order. The circular variance, $1-R_1$, replaces the linear variance $\mathrm{Var}(\theta)$, which is ill-defined on a periodic domain. To understand why $1-R_1$ acts as a variance, notice that the average $N^{-1}\sum_i e^{i\theta_i}$ always lies inside the unit disc. If all angles coincide, the unit vectors add coherently, the average sits exactly on the unit circle, and $1-R_1=0$. As the sample spreads out, the unit vectors cancel each other, the average retreats toward the center of the disc, and $1-R_1$ grows toward $1$ (its value for a uniform distribution). The distance of the resultant from the unit circle thus plays the exact same role as the spread about the mean on the real line. For a highly concentrated sample, this identification is quantitative: expanding $e^{i(\theta_i-\bar\theta)}$ about the mean direction $\bar\theta$ yields $1-R_1 \approx \tfrac{1}{2}N^{-1}\sum_i(\theta_i-\bar\theta)^2$, which is half the linear variance.

Calamitic molecules, however, are axes rather than directed vectors. Because $\theta$ and $\theta+\pi$ represent the exact same physical configuration, their orientations actually live on the real projective line $\mathbb{RP}^1$. Consequently, any polar statistic will spuriously average to zero in a head--tail symmetric nematic phase. The standard workaround is the angle-doubling map $\theta\mapsto 2\theta$, which smoothly identifies $\mathbb{RP}^1$ with $\mathbb{S}^1$. Thus, the order parameter we monitor throughout our tests is the circular variance of these doubled angles:
\begin{equation}\label{eq:circvar}
	\sigma^2(\theta) = 1-R_2, \qquad R_2 = \Big|\frac{1}{N}\sum_{i=1}^N e^{2i\theta_i}\Big|.
\end{equation}
This variance equals $1$ in a perfectly isotropic phase and reduces to $0$ as the rods perfectly align along a common director. Note that $R_2$, the mean resultant length of the doubled angles, exactly matches the standard two-dimensional nematic scalar order parameter (i.e., twice the positive eigenvalue of the $Q$-tensor $Q = \langle \vec{\nu}\otimes\vec{\nu}\rangle - \tfrac12 \mathrm{Id}$ \rev{\cite[Chapter~2]{degennes}}). Analogous geometric mappings exist for the sphere and higher-dimensional projective spaces \cite{mardiaJupp}. \reviv{By the $v_x$ and $v_y$ marginal of $f$ we mean the density obtained by integrating $f$ in all the variables except $v_x$ or $v_y$ respectively, and likewise for the other marginals, while the $(\theta,\omega)$ density is the marginal in the pair $(\theta,\omega)$.}

\subsection{Test 1: Thermalization and the \texorpdfstring{$H$}{H}-theorem}
\label{sec:thermalization}
We begin by deactivating the mean-field potential ($\mathcal{V}=0$) in \eqref{eq:Boltzmann2D_calamitic}. The $H$-theorem predicts that the system described by \eqref{eq:Boltzmann2D_calamitic} relaxes to the Maxwellian distribution \cite[Section~3.4]{carrillo2025kinetic}:
\begin{equation}\label{eq:maxwellian_local}
	f^\infty(\vec{v},\theta,\omega) = \rho^{\infty}(\theta) \frac{\sqrt{I_3}}{(2\pi T)^{3/2}} \exp\left( - \frac{|\vec{v} - \vec{U}|^2}{2T} \right) \exp\left( - \frac{I_3( \omega - \Omega)^2}{2T} \right).
\end{equation}
Here $T$, $\vec{U}$, $\Omega$, and $\rho^\infty(\theta)$ are the equilibrium temperature, mean velocity, mean angular velocity, and mass fraction, respectively. The initial datum's energy, momentum, and angular momentum strictly determine the first three parameters, which remain independent of $\theta$. Because the collision operator couples all orientations $\theta_*$, the collision integral only vanishes once these intensive parameters fully equilibrate across the entire angular domain. Only the mass fraction $\rho^\infty(\theta)$ remains an arbitrary function of $\theta$, as collisions do not exchange orientations.
We take a factorised initial datum $f^0(\vec{v}, \theta, \omega) = f^0(\vec{v})f^0(\theta)f^0(\omega)$ with
\begin{equation}\label{eq:ic}
f^0(\vec{v}) \propto \chi_{[-1/2,\,1/2]^2}(\vec{v}), \quad f^0(\theta) \propto \chi_{[-\pi,\,\pi]}(\theta), \quad f^0(\omega) \propto \chi_{[-1/4,\,1/4]}(\omega),
\end{equation}
\rev{This provides a non-equilibrium initial datum, far from any Gaussian, whose relaxation is easy to visualize, and the unequal translational and rotational temperatures, $T_{\mathrm{tr}}=1/12$ and $T_{\mathrm{rot}}=1/48$, let the collisions redistribute the energy between \revii{the translational and the rotational components of the temperature}.}

We test two variants of the scheme: one with the transport step active, and one without. In both runs, the system successfully relaxes to \cref{eq:maxwellian_local}. Without a mean-field potential, the orientational distribution stays uniform. The two variants yield indistinguishable results, precisely as expected. When the potential is zero, the transport substep merely induces free rotation, leaving the uniform orientational distribution and all its marginals completely unchanged. We therefore plot only a single set of results (the variant without transport) in \cref{fig:test01}. This illustrates the numerical counterpart of the $H$-theorem \cite{carrillo2025kinetic}: relaxation is strictly monotone and targets the Maxwellian.

\begin{figure}[tb]
	\centering
	\subfloat[$v_y$ marginal, $t=0$]{\includegraphics[width=0.32\textwidth]{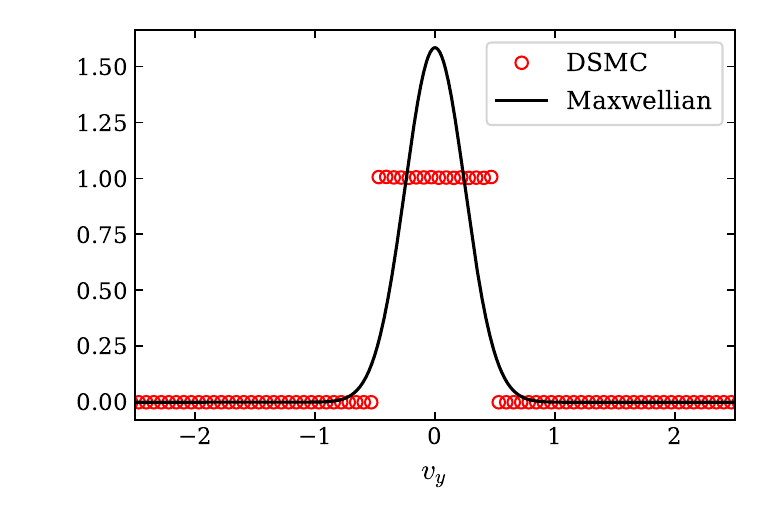}}\hfill
	\subfloat[$v_y$ marginal, $t=5$]{\includegraphics[width=0.32\textwidth]{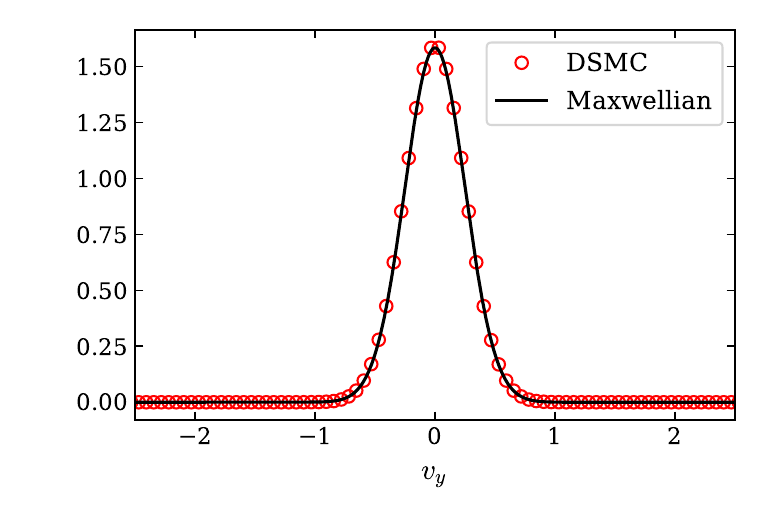}}\hfill
	\subfloat[$v_y$ marginal, $t=50$]{\includegraphics[width=0.32\textwidth]{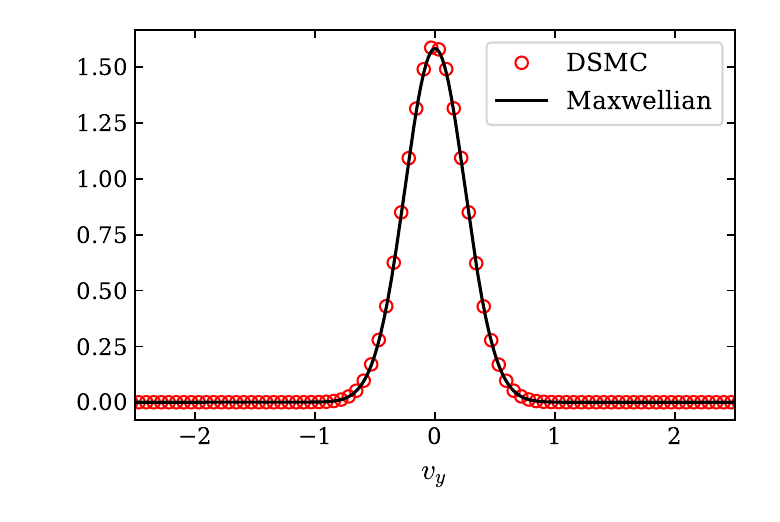}}\\
	\subfloat[$v_x$ marginal, $t=50$]{\includegraphics[width=0.32\textwidth]{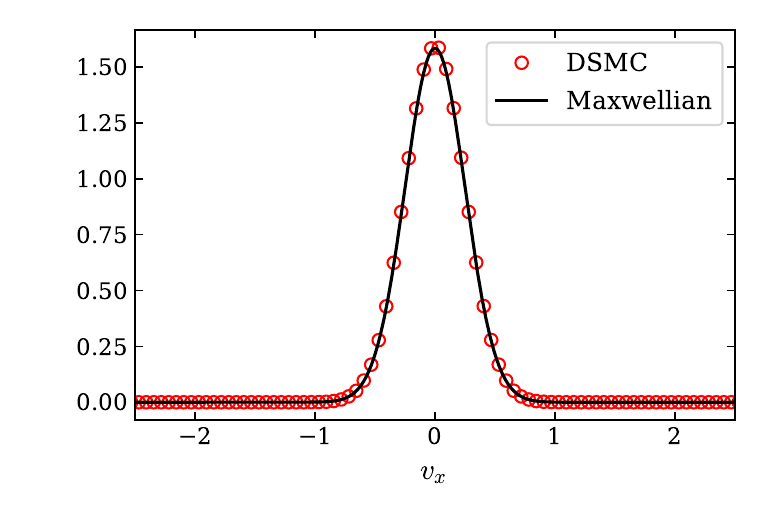}}\hfill
	\subfloat[$\omega$ marginal, $t=50$]{\includegraphics[width=0.32\textwidth]{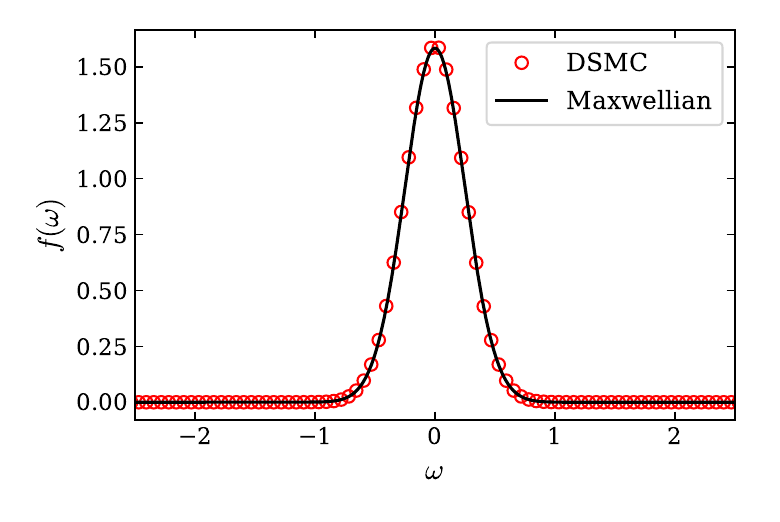}}\hfill
	\subfloat[$\theta$ marginal, $t=50$]{\includegraphics[width=0.32\textwidth]{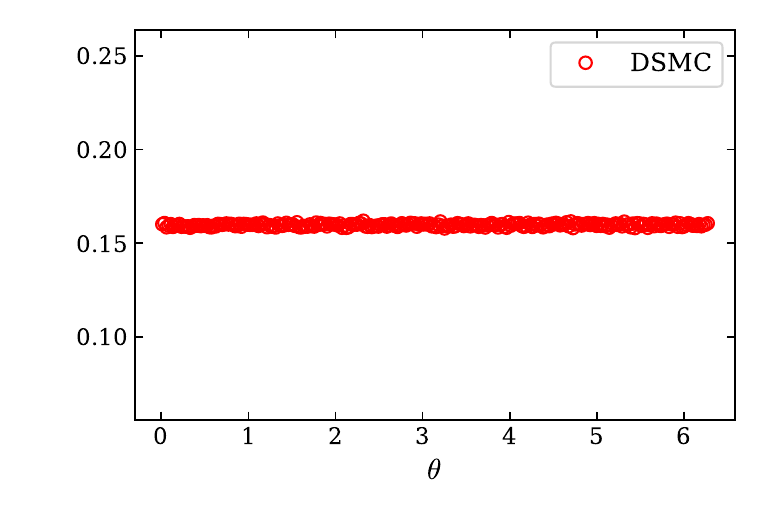}}\\
	\subfloat[temperature history]{\includegraphics[width=0.4\textwidth]{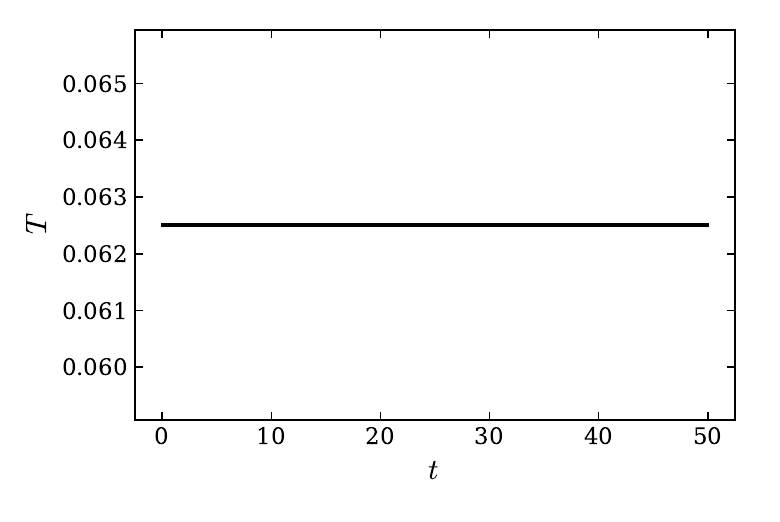}}
	\caption{Test~1. \revjac{We simulate \eqref{eq:Boltzmann2D_calamitic} with $\mathcal{V}=0$, to show thermalization with zero potential and elastic collisions (no transport).} The top row shows three snapshots of the $v_y$ marginal, which relaxes from the \rev{flat-top} initial datum \cref{eq:ic} to the Maxwellian. The middle row shows the remaining one-dimensional marginals at the final time, $t=50$. The markers are the DSMC solution and the solid line the Maxwellian, while the orientational marginal in panel (f) remains uniform. The bottom row shows the time evolution of the temperature. \reviv{$N=10^7$, $\Delta t=0.05$, $1/\tau=10$. The velocity and angular-velocity marginals are shown in detail on $[-2.5,2.5]$.}}
	\label{fig:test01}
\end{figure}

\subsection{Test 2: Mean-field alignment}
\label{sec:alignment}
Next, we activate the mean-field potential \revjac{in \cref{eq:Boltzmann2D_calamitic}} to study the emergence of orientational order. We test both the quadratic force \cref{eq:quadratic_force} (in the frictionless case, $\beta=0$) and the Kuramoto force \cref{eq:cosine_force}. Starting from a uniform orientational distribution, the rods successfully align around the mean director. As shown in \cref{fig:alignment}, the orientational distribution concentrates, the circular variance reduces, and the temperature undergoes a transient phase before settling. 

Unlike the pure thermalization case, the dynamics are now driven by a self-consistent force. Consequently, the relaxation is strictly non-monotone, marked by damped oscillations in both temperature and circular variance.
While both potentials drive the system toward the same qualitative endpoint, a single, highly concentrated cluster in the $(\theta,\omega)$ plane, they take distinct transient paths. The quadratic force, which grows unboundedly with angular separation, condenses the distribution much faster and triggers wider temperature oscillations. The Kuramoto force, by contrast, saturates at a separation of $\pi/2$, leading to a gentler condensation. The oscillations themselves stem purely from mean-field effects (absent in \cref{sec:thermalization}): the nucleating cluster overshoots the mean director and gets pulled back. This periodic exchange between potential and kinetic energy manifests as damped oscillations in the macroscopic temperature and kinetic energy, with the collisions providing the necessary damping. Ultimately, the circular variance settles at a small but non-zero value, reflecting the thermal broadening of the cluster at equilibrium. \revjac{The slow increase of the kinetic energy after the plateau under the quadratic force is a numerical error. The force jumps at the cut of the chart, see \cref{sec:quadratic}, and the symplectic integrator loses its energy behaviour for the molecules that cross the jump, each crossing feeding energy into the system. The Kuramoto force is smooth on the circle and its kinetic energy shows no drift.}

\begin{figure}[!tb]
	\centering
	\raisebox{1em}{\begin{minipage}[c]{0.035\textwidth}\centering\rotatebox{90}{\footnotesize Quadratic potential}\end{minipage}}\hfill
	\begin{minipage}[c]{0.95\textwidth}
		\centering
		\subfloat[$(\theta,\omega)$ density, $t=0$]{\includegraphics[width=0.32\textwidth]{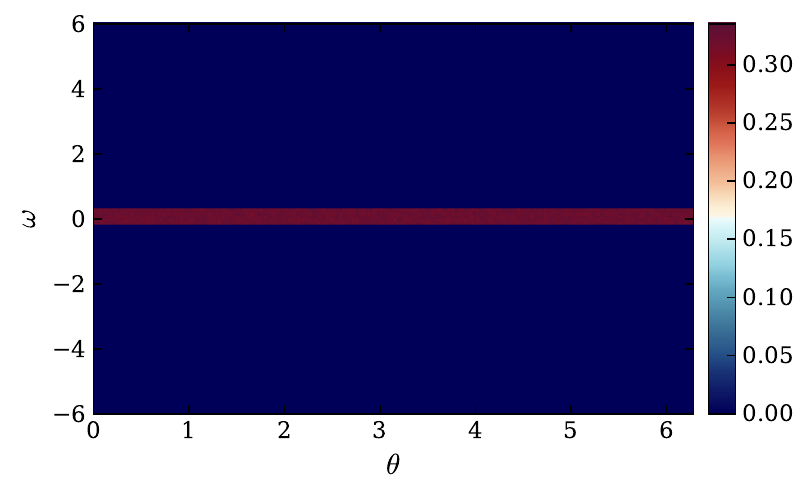}}\hfill
		\subfloat[$(\theta,\omega)$ density, $t=5$]{\includegraphics[width=0.32\textwidth]{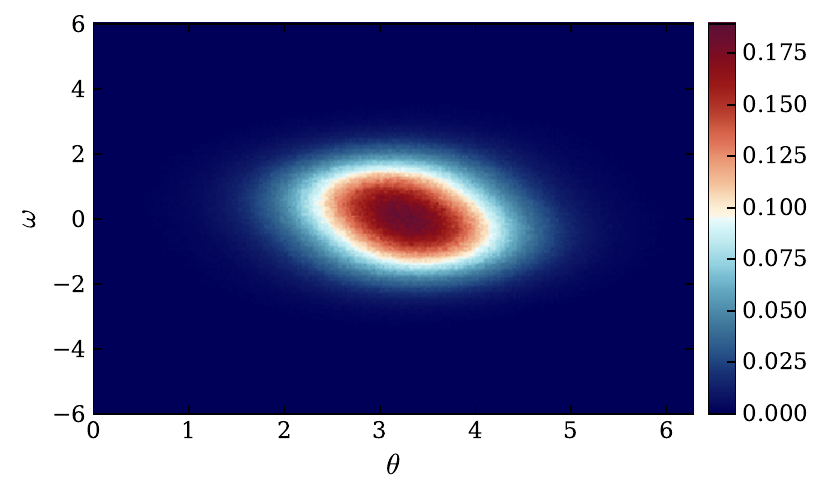}}\hfill
		\subfloat[$(\theta,\omega)$ density, $t=10$]{\includegraphics[width=0.32\textwidth]{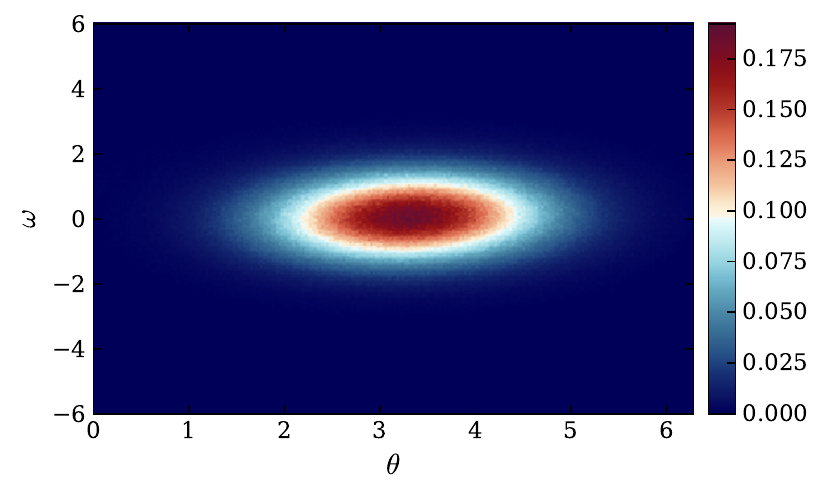}}
	\end{minipage}\\
	\raisebox{1em}{\begin{minipage}[c]{0.035\textwidth}\centering\rotatebox{90}{\footnotesize Kuramoto potential}\end{minipage}}\hfill
	\begin{minipage}[c]{0.95\textwidth}
		\centering
		\subfloat[$(\theta,\omega)$ density, $t=0$]{\includegraphics[width=0.32\textwidth]{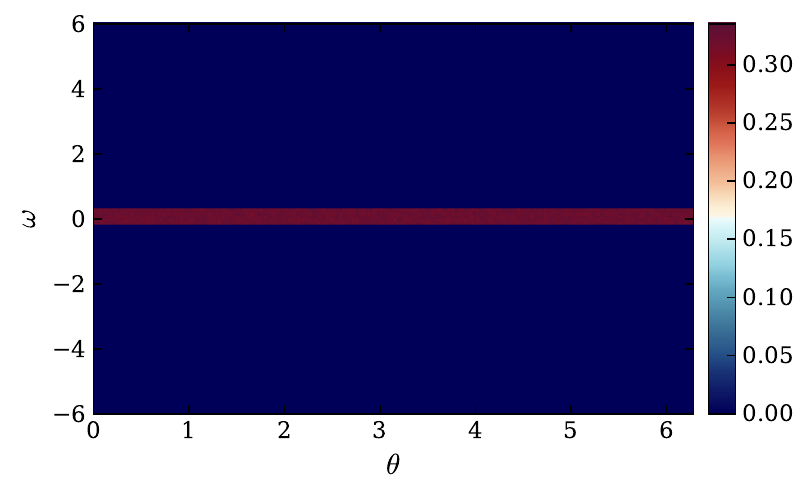}}\hfill
		\subfloat[$(\theta,\omega)$ density, $t=5$]{\includegraphics[width=0.32\textwidth]{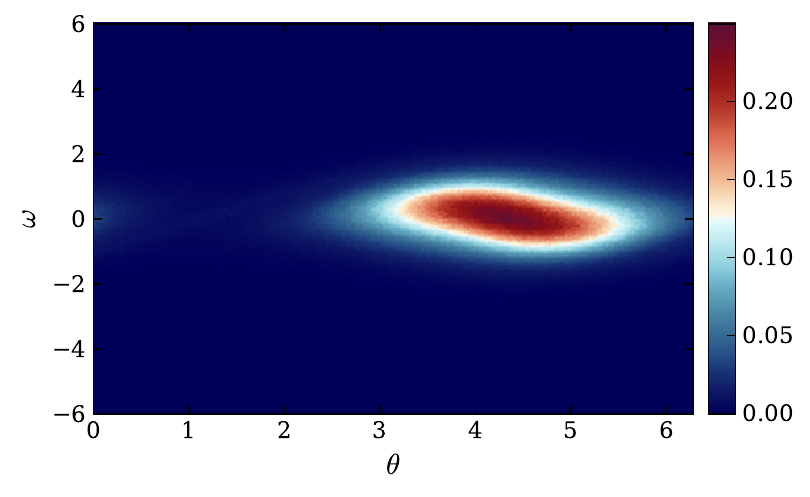}}\hfill
		\subfloat[$(\theta,\omega)$ density, $t=10$]{\includegraphics[width=0.32\textwidth]{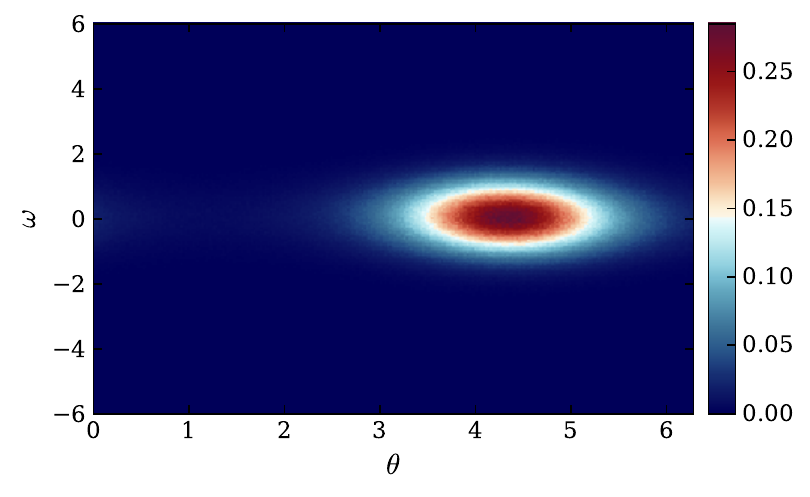}}
	\end{minipage}\\
	\raisebox{1em}{\begin{minipage}[c]{0.035\textwidth}\centering\rotatebox{90}{\footnotesize Quadratic potential}\end{minipage}}\hfill
	\begin{minipage}[c]{0.95\textwidth}
		\centering
		\subfloat[temperature]{\includegraphics[width=0.32\textwidth]{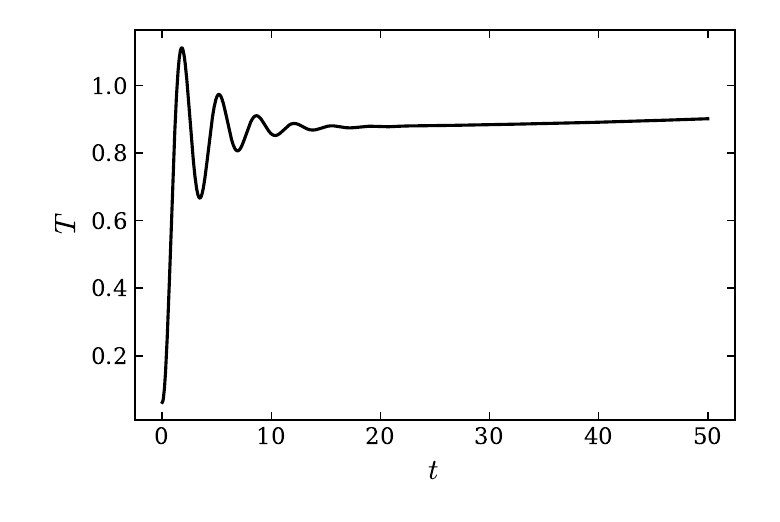}}\hfill
		\subfloat[circular variance]{\includegraphics[width=0.32\textwidth]{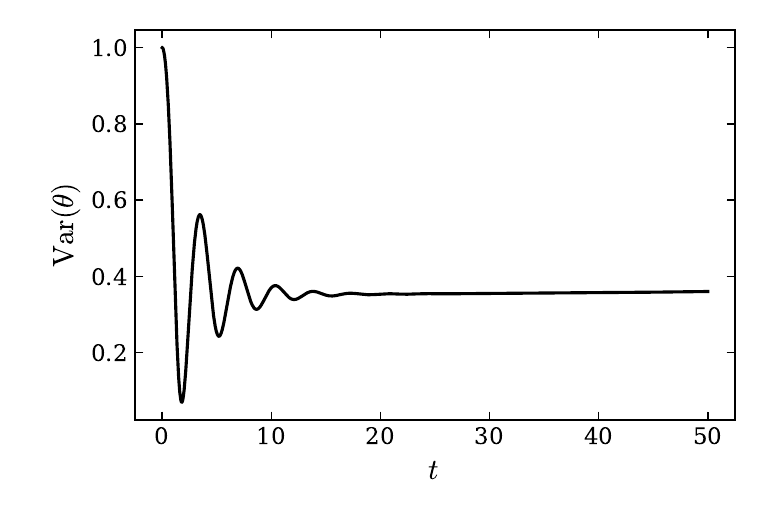}}\hfill
		\subfloat[kinetic energy]{\includegraphics[width=0.32\textwidth]{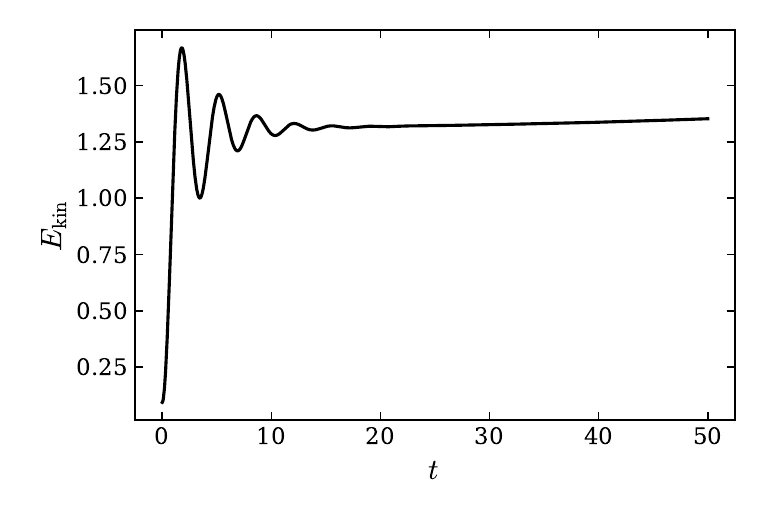}}
	\end{minipage}\\
	\raisebox{1em}{\begin{minipage}[c]{0.035\textwidth}\centering\rotatebox{90}{\footnotesize Kuramoto potential}\end{minipage}}\hfill
	\begin{minipage}[c]{0.95\textwidth}
		\centering
		\subfloat[temperature]{\includegraphics[width=0.32\textwidth]{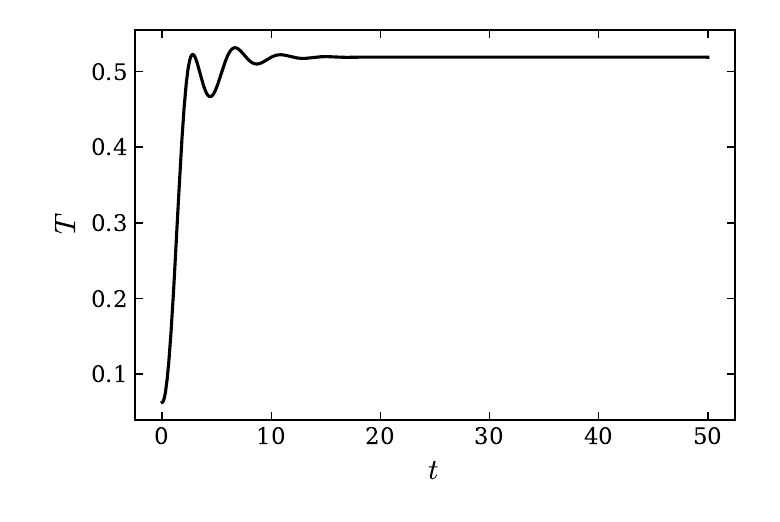}}\hfill
		\subfloat[circular variance]{\includegraphics[width=0.32\textwidth]{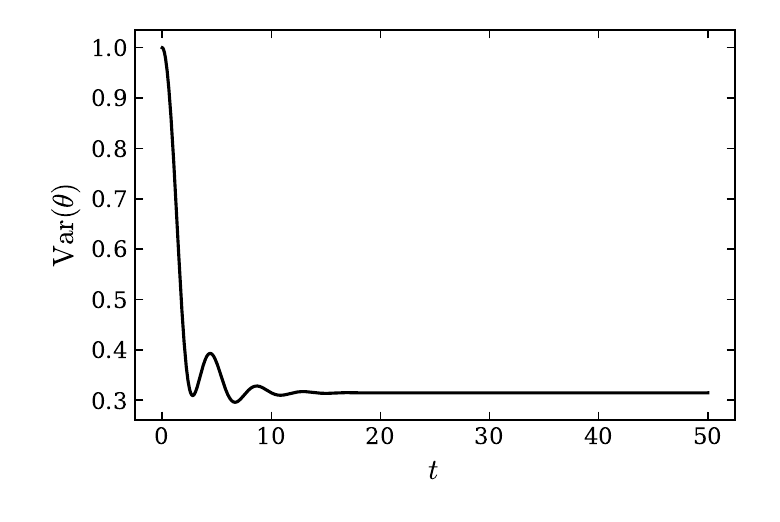}}\hfill
		\subfloat[kinetic energy]{\includegraphics[width=0.32\textwidth]{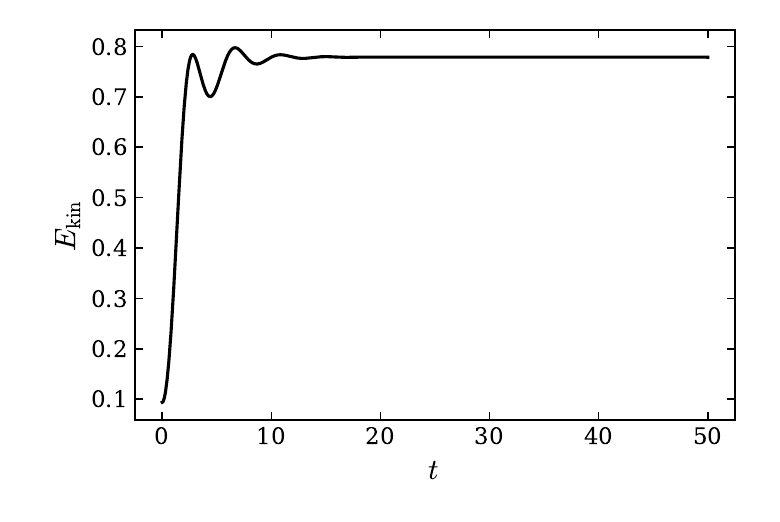}}
	\end{minipage}
	\caption{Test~2. \revjac{We simulate \eqref{eq:Boltzmann2D_calamitic} with the quadratic force \cref{eq:quadratic_force} ($\alpha=1$, $\beta=0$) in panels (a)--(c), (g)--(i) and with the Kuramoto force \cref{eq:cosine_force} in panels (d)--(f), (j)--(l), to show the mean-field alignment of the rods about the mean director.} Each of the first two rows shows three snapshots of the $(\theta,\omega)$ density during the initial transient, as it concentrates about the mean director. The last two rows show the time evolution of the temperature, the circular variance $\sigma^2(\theta)$ of \cref{eq:circvar}, and the kinetic energy, whose damped oscillations are the signature of the self-consistent mean-field coupling. As the rods align, interaction energy is converted into kinetic energy and the kinetic energy rises to a plateau. \reviv{$N=10^7$, $\Delta t=0.05$, $1/\tau=4$. The $(\theta,\omega)$ densities are shown on $\omega\in[-6,6]$.}}
	\label{fig:alignment}
\end{figure}

A more delicate test involves nucleating order from a nearly isotropic state. \revii{To do this, we perturb the uniform initial datum along its first harmonic:
\begin{equation}\label{eq:ic_perturbed}
	f^0(\theta) = \frac{1}{2\pi}\Big(1+\varepsilon\cos\left(\theta-\theta_0\right)\Big), \qquad \varepsilon = 10^{-1}, \qquad \theta_0 = 0.3.
\end{equation}
This $\theta$-dependent perturbation subtly breaks symmetry and singles out a preferred orientation $\theta_0$.} Under the Kuramoto dynamics, the distribution condenses precisely around this preferred director (see \cref{fig:nucleation}). The mean-field force catches the small initial anisotropy and amplifies it until a macroscopic ordered state emerges. This test is numerically delicate because the initial signal sits at the perturbation scale $\varepsilon$, making the early-time mean-field force extremely weak. Nevertheless, the growth remains robustly exponential, and the circular variance decreases appreciably.

\begin{figure}[tb]
	\centering
	\subfloat[$(\theta,\omega)$ density, $t=0$]{\includegraphics[width=0.32\textwidth]{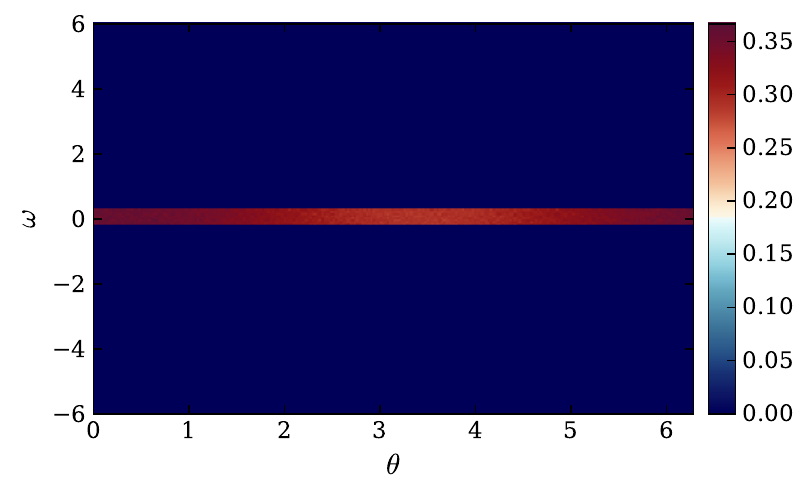}}\hfill
	\subfloat[$(\theta,\omega)$ density, $t=5$]{\includegraphics[width=0.32\textwidth]{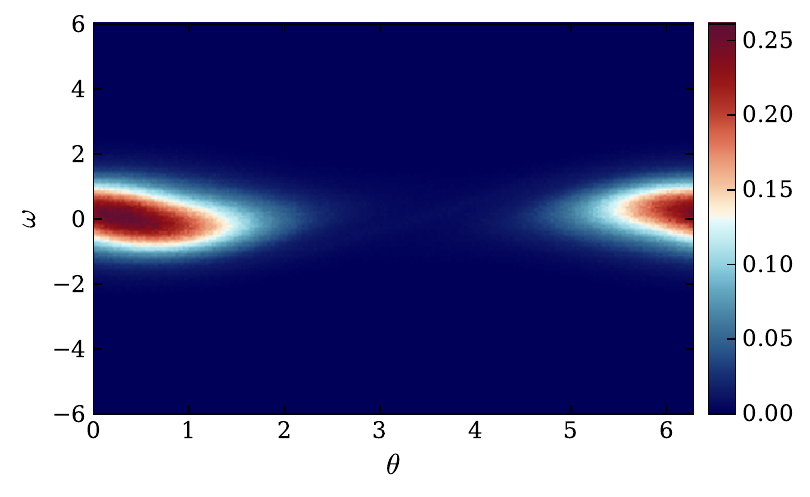}}\hfill
	\subfloat[$(\theta,\omega)$ density, $t=10$]{\includegraphics[width=0.32\textwidth]{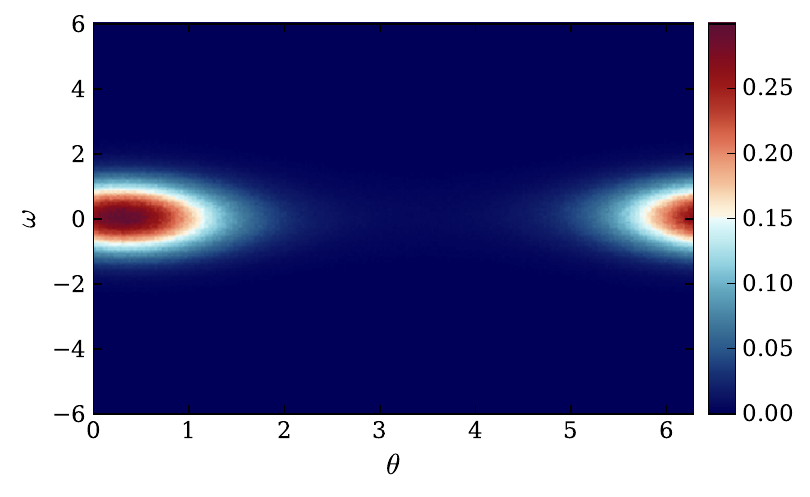}}\\
	\subfloat[$\theta$ marginal, $t=0$]{\includegraphics[width=0.32\textwidth]{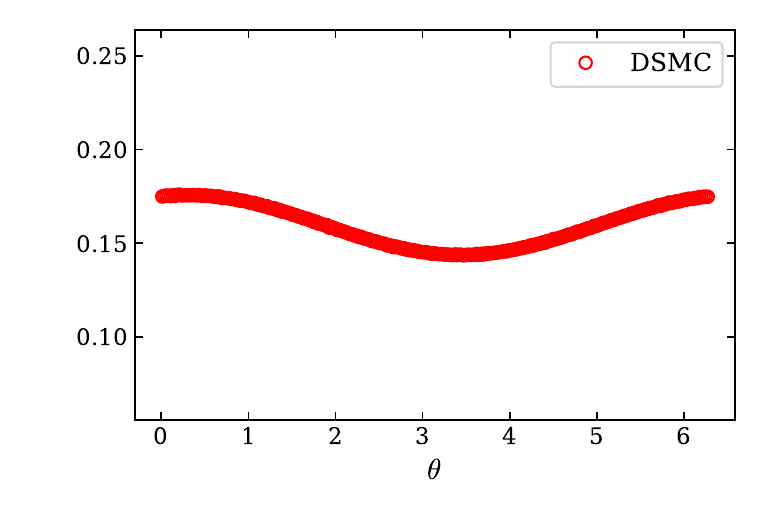}}\hfill
	\subfloat[$\theta$ marginal, $t=5$]{\includegraphics[width=0.32\textwidth]{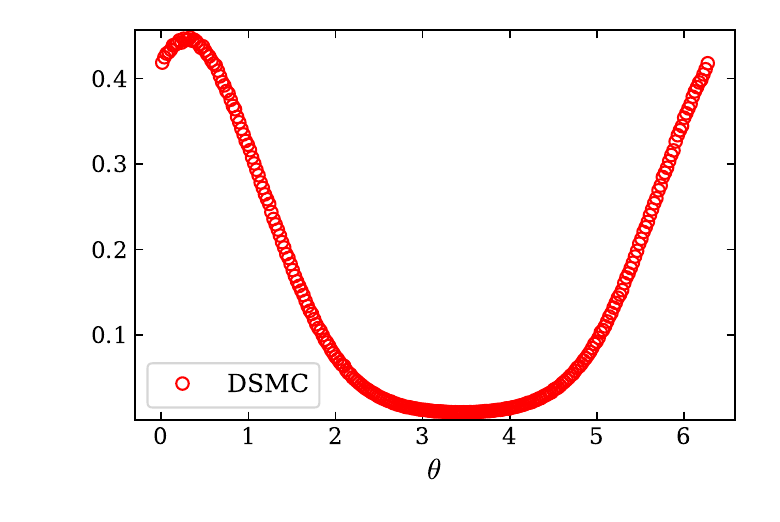}}\hfill
	\subfloat[$\theta$ marginal, $t=10$]{\includegraphics[width=0.32\textwidth]{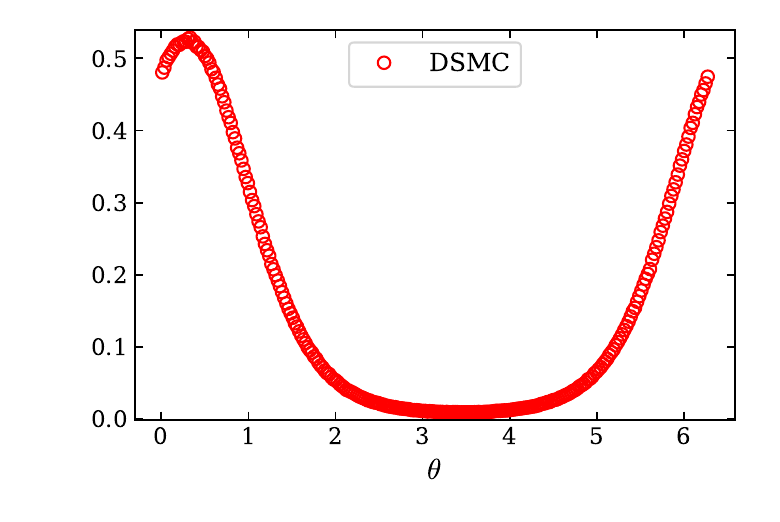}}\\
	\subfloat[temperature]{\includegraphics[width=0.32\textwidth]{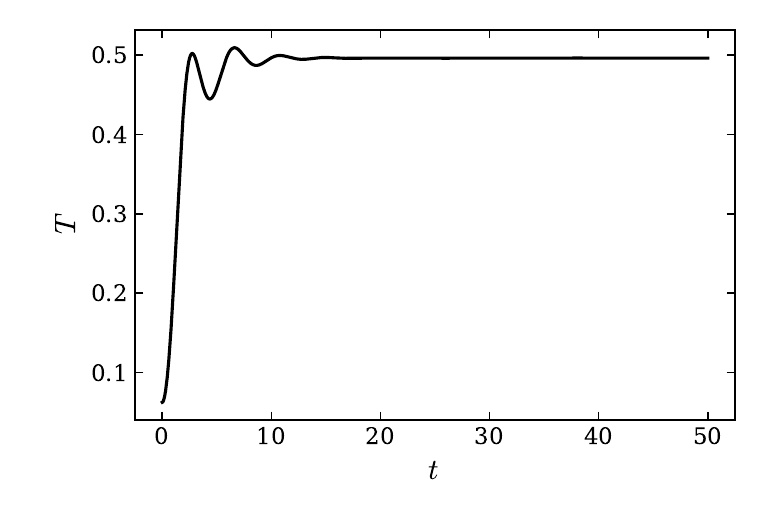}}\hfill
	\subfloat[circular variance]{\includegraphics[width=0.32\textwidth]{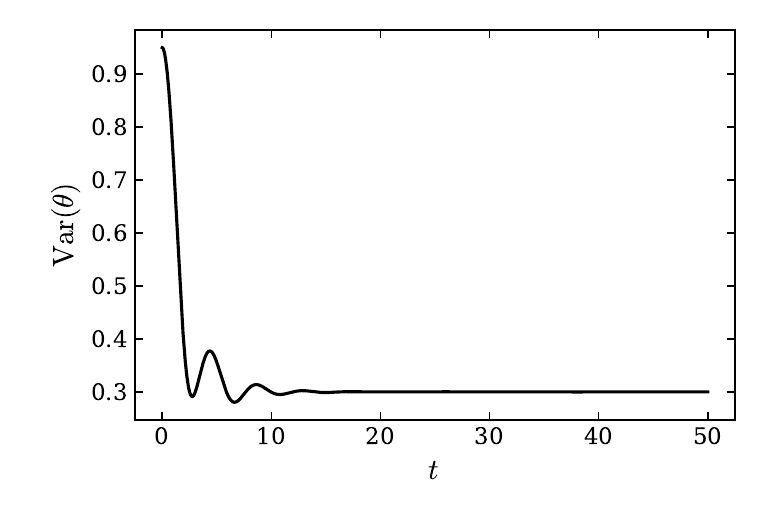}}\hfill
	\subfloat[kinetic energy]{\includegraphics[width=0.32\textwidth]{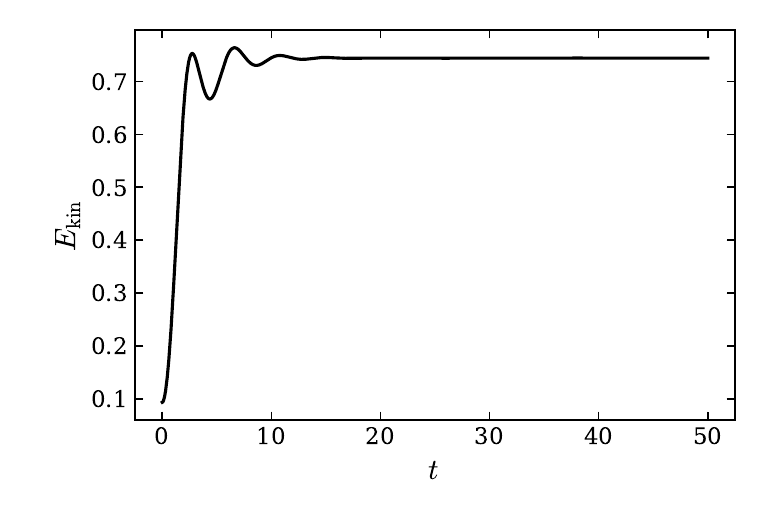}}
	\caption{Test~2. \revjac{We simulate \eqref{eq:Boltzmann2D_calamitic} with the Kuramoto force \cref{eq:cosine_force} and the perturbed isotropic initial datum \cref{eq:ic_perturbed}, to show the nucleation of order from a small anisotropy.} The top row shows three snapshots of the $(\theta,\omega)$ density during the initial transient, and the middle row the corresponding $\theta$ marginals. In panel (b) the nucleating cluster is still draining the residual isotropic band, which has disappeared by $t=10$. The bottom row shows the time evolution of the temperature, the circular variance, and the kinetic energy. The small initial anisotropy is amplified into a macroscopic ordered state. \reviv{$N=10^7$, $\Delta t=0.05$, $1/\tau=4$. The $(\theta,\omega)$ densities are shown on $\omega\in[-6,6]$.}}
	\label{fig:nucleation}
\end{figure}

\subsection{Test 3: The Onsager potential: energy conversion and instabilities}
\label{sec:onsager}
We next examine the physically motivated Onsager potential \cref{eq:onsager} \revjac{in \cref{eq:Boltzmann2D_calamitic}}, driven by the full mean-field force \cite{carrillo2025kinetic}. In a microcanonical setting (without a thermostat), the total energy \cref{eq:totalenergy} is rigorously conserved (\cref{prop:energy}). As nematic order crystallizes, the system lowers its interaction energy. By conservation, the kinetic energy---and thus the temperature---must rise by exactly the same amount. \Cref{fig:onsager} plots this energy conversion alongside the total energy evolution, while the top-row snapshots capture the morphological emergence of order. Initially, the isotropic state forms a flat band in phase space. The unstable nematic mode causes this band to buckle and roll up into two distinct bumps separated by half a period in $\theta$. These two bumps naturally represent the single nematic director, perfectly respecting head--tail symmetry. As noted earlier, the fully discrete scheme does not conserve total energy perfectly; over the entire simulation horizon, it drifts upwards by roughly $0.6\%$. Crucially, however, this numerical drift remains more than an order of magnitude smaller than the physical energy exchanged between the interaction and kinetic components.

\begin{figure}[tb]
	\centering
	\subfloat[$(\theta,\omega)$ density, $t=2.5$]{\includegraphics[width=0.32\textwidth]{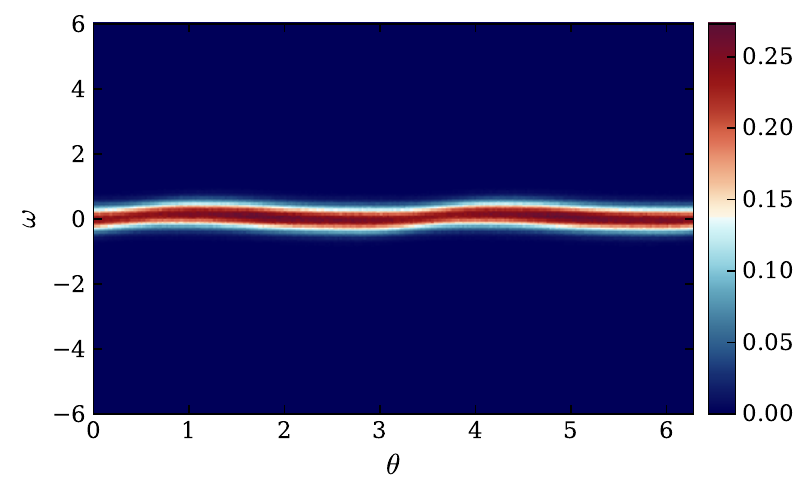}}\hfill
	\subfloat[$(\theta,\omega)$ density, $t=5$]{\includegraphics[width=0.32\textwidth]{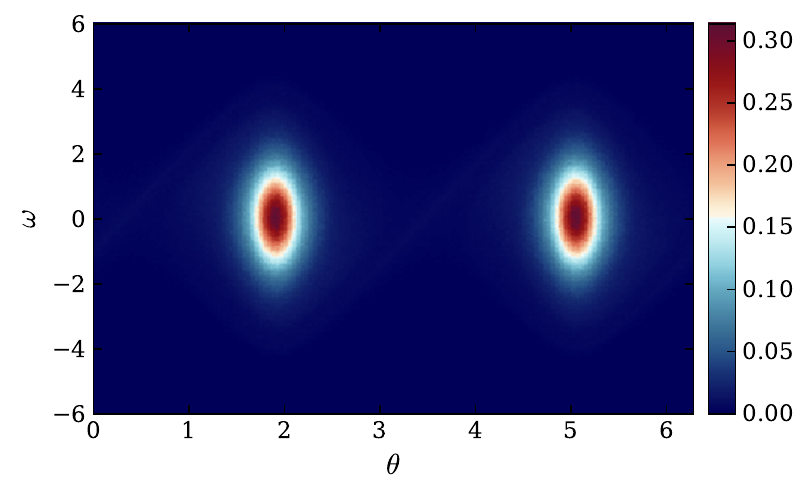}}\hfill
	\subfloat[$(\theta,\omega)$ density, $t=10$]{\includegraphics[width=0.32\textwidth]{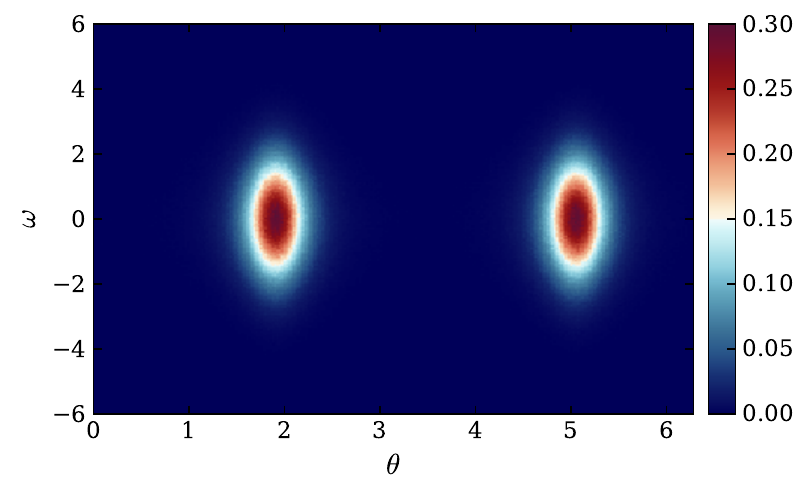}}\\
	\subfloat[temperature]{\includegraphics[width=0.24\textwidth]{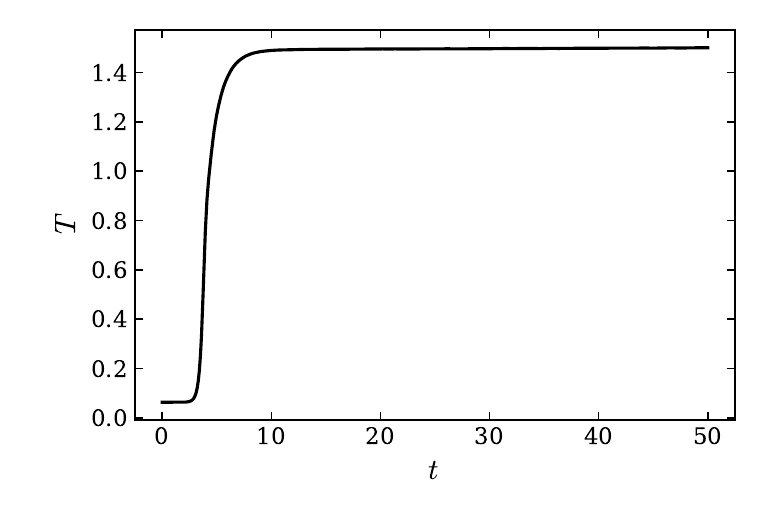}}\hfill
	\subfloat[circular variance]{\includegraphics[width=0.24\textwidth]{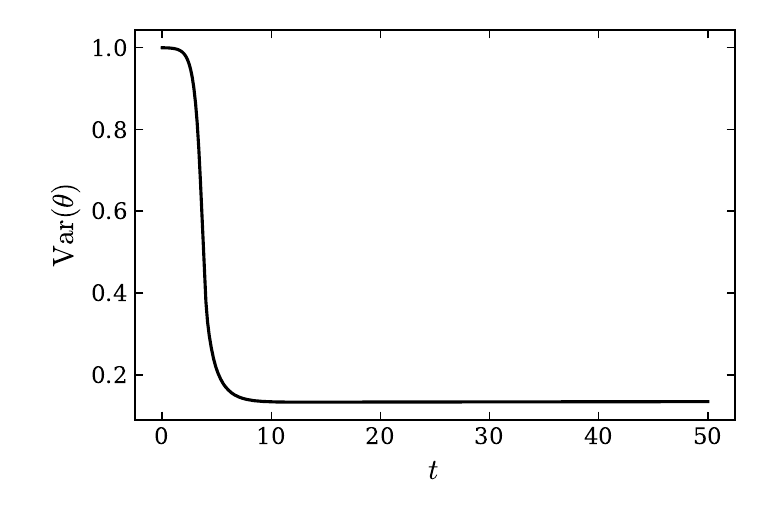}}\hfill
	\subfloat[interaction energy]{\includegraphics[width=0.24\textwidth]{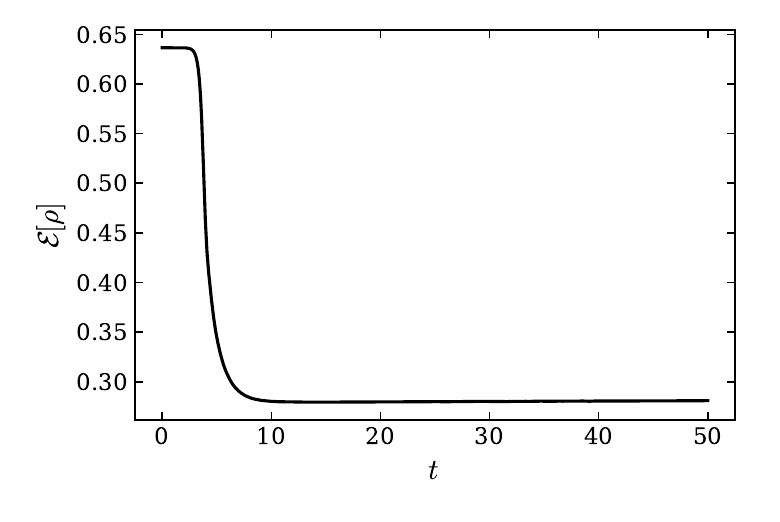}}\hfill
	\subfloat[total energy]{\includegraphics[width=0.24\textwidth]{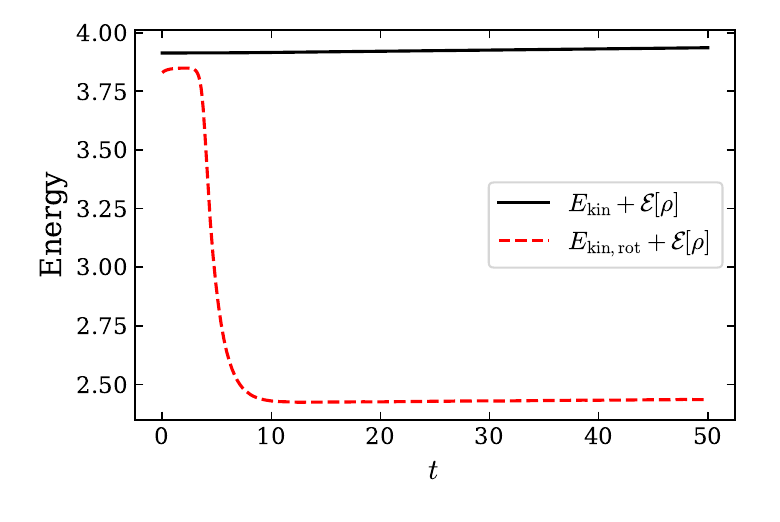}}
	\caption{Test~3. \revjac{We simulate \eqref{eq:Boltzmann2D_calamitic} with the Onsager potential \cref{eq:onsager} in the microcanonical ensemble, to show the emergence of nematic order and the conversion of interaction energy into kinetic energy.} The top row shows three early snapshots of the $(\theta,\omega)$ density. The initially flat band buckles in panel (a), rolls up into the two head--tail symmetric nematic bumps in panel (b), and by $t=10$ the bumps have relaxed to their equilibrium shape in panel (c). The bottom row shows the time evolution of the scalar observables. The temperature rises as the interaction energy in panel (f) is converted into kinetic energy. In panel (g) the total energy (solid line) drifts by about $0.6\%$ of its value over the whole run. The sum of rotational kinetic and interaction energy (dashed line) is not a conserved quantity, since the collisions exchange rotational and translational kinetic energy. \reviv{$L=\sqrt{12}$, $N=10^7$, $\Delta t=0.05$, $1/\tau=4$, $N_\theta=256$. The $(\theta,\omega)$ densities are shown on $\omega\in[-6,6]$.}}
	\label{fig:onsager}
\end{figure}

To isolate the transport dynamics, we rerun the Onsager potential with collisions deactivated. As \cref{fig:filamentation} shows, the collisionless mean-field flow spontaneously drives a striking filamentation of phase space, strongly resembling instabilities found in plasma physics. Specifically, the snapshots mirror the diocotron instability of non-neutral plasmas, where an annular electron column rolls up into a chain of vortices shrouded by ever-finer filaments \cite{levy1965,knauer1966,driscollFine}. This is precisely the mechanism driving orientational concentration in the collisionless regime. 

Initially, the distribution forms a uniform band in the $(\theta,\omega)$ plane. The mean-field force shears this band, causing particles at different angular velocities $\omega$ to rotate at different rates. Consequently, any perturbation to the band winds rather than relaxes. This winding rolls the band up into the chain of vortices visible in the second snapshot. Each vortex traps orientational mass and vigorously stirs it toward its center, while the interstitial material is stretched into tight filaments winding around the vortices. It is exactly this macroscopic winding that aggregates the orientational marginal into the nematic bumps. Coarse observables (like the circular variance and interaction energy) thus converge, even though the phase-space density never relaxes pointwise---a direct analogue to phase mixing in the Vlasov--Poisson equation. Lacking collisions to dissipate them, these filaments grow increasingly thin until they fall below the resolution of our density reconstruction. Here, the particle method excels: it transports these filaments without artificial diffusion, whereas a purely Eulerian mesh-based method would artificially smear them once they drop below the grid scale. Throughout this collisionless run, the total energy drifts by only $0.4\%$, remaining an order of magnitude smaller than the physical energy exchange.

\begin{figure}[tb]
	\centering
	\subfloat[$(\theta,\omega)$ density, $t=5$]{\includegraphics[width=0.32\textwidth]{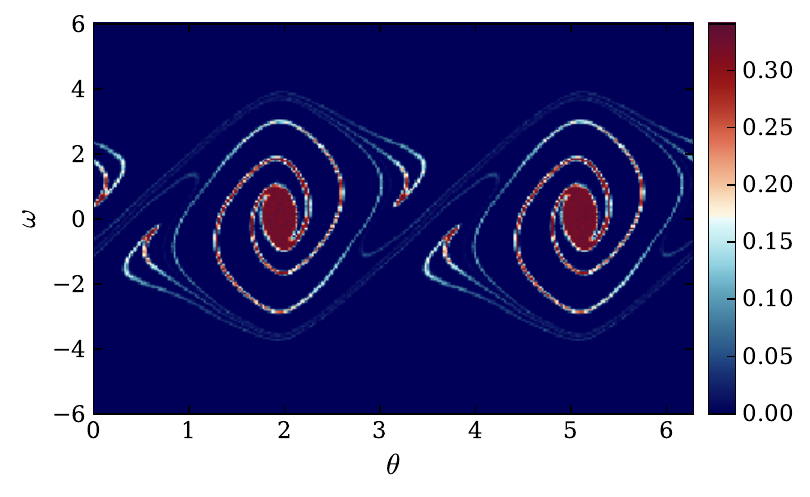}}\hfill
	\subfloat[$(\theta,\omega)$ density, $t=15$]{\includegraphics[width=0.32\textwidth]{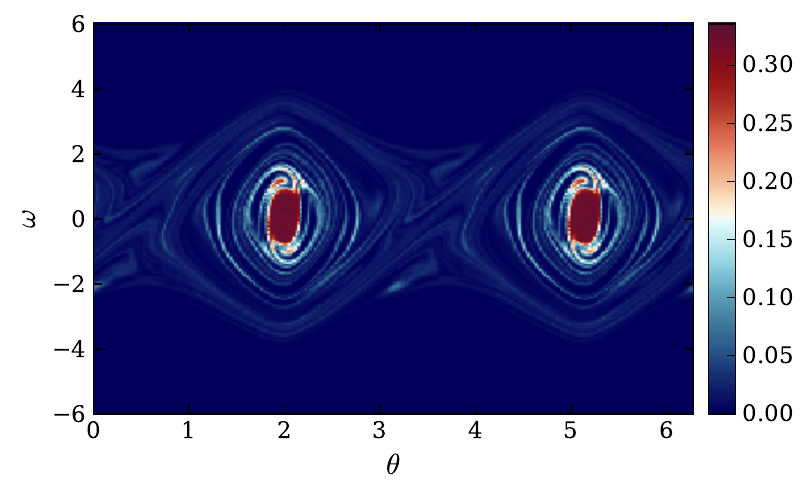}}\hfill
	\subfloat[$(\theta,\omega)$ density, $t=35$]{\includegraphics[width=0.32\textwidth]{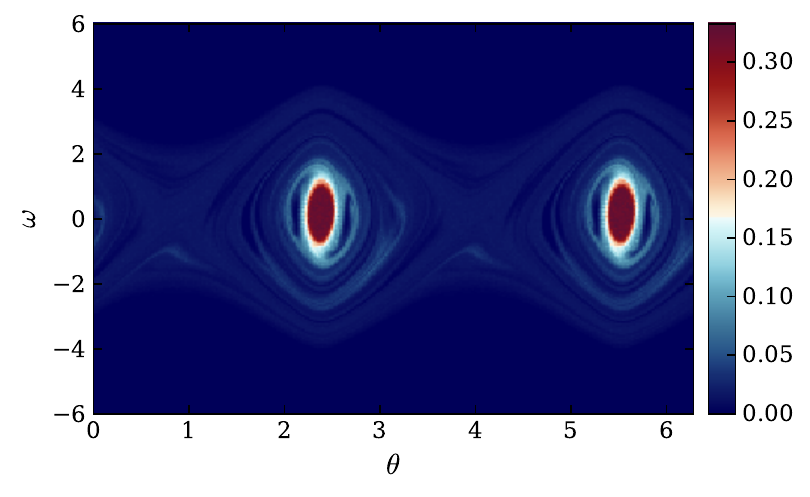}}\\
	\subfloat[total energy]{\includegraphics[width=0.45\textwidth]{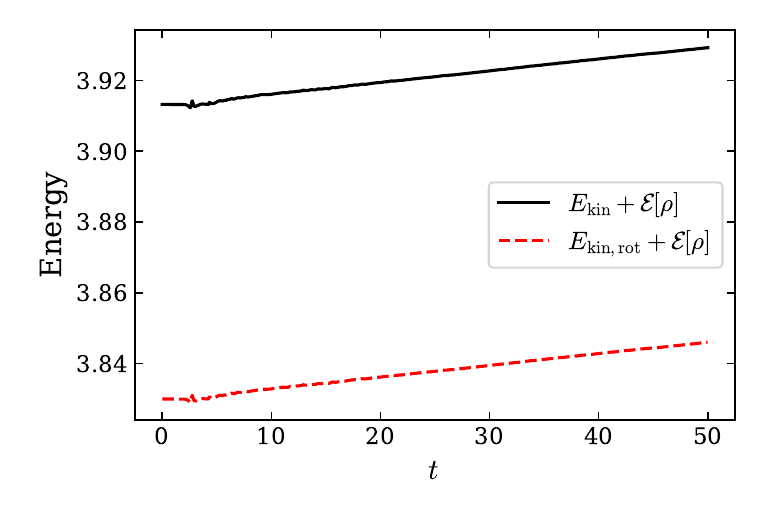}}
	\caption{Test~3. \revjac{We simulate \eqref{eq:Boltzmann2D_calamitic} with the Onsager potential \cref{eq:onsager} and without collisions, to show the filamentation of phase space.} The $(\theta,\omega)$ density at three successive times shows \revjac{the band rolling up} into a chain of vortices, resembling the diocotron instability of non-neutral plasmas \cite{levy1965,driscollFine}. The total energy in panel (d) (solid line) drifts by about $0.4\%$ of its value over the whole run. In the absence of collisions the sum of rotational kinetic and interaction energy (dashed line) behaves in the same way, as no energy is exchanged with the translational degrees of freedom. \reviv{Parameters as in \cref{fig:onsager}. The $(\theta,\omega)$ densities are shown on $\omega\in[-6,6]$.}}
	\label{fig:filamentation}
\end{figure}

\subsection{Test 4: The isotropic--nematic transition}
\label{sec:transition}
Finally, we apply the thermostat from \cref{sec:thermostat} \revjac{to \cref{eq:Boltzmann2D_calamitic} with the Onsager potential} to investigate the phase diagram of the Onsager model as the temperature is varied. \reviv{By \cref{prop:critical} the isotropic state loses its stability at the temperature $T_c$ of \cref{eq:Tc}.}

In terms of the dimensionless coupling $\alpha = L^2/T$, the stability condition \cref{eq:Tc} reads $\alpha < \alpha_c = 3\pi/2$. For $L=\sqrt{12}$ the critical temperature is $T_c = 8/\pi \approx 2.55$. Varying the bath temperature $T_{\mathrm{bath}}$ across $T_c$ and recording \revjac{the nematic order parameter $R_2=1-\sigma^2(\theta)$ of} \cref{eq:circvar} at equilibrium, in \cref{fig:phasediagram} we observe a continuous transition from the isotropic phase (\revjac{$R_2\approx 0$} for $T_{\mathrm{bath}}>T_c$) to the nematic phase (\revjac{$R_2\to 1$} for $T_{\mathrm{bath}}<T_c$). For each value of $T_{\mathrm{bath}}$ we run a new simulation with $N=4\times 10^5$ particles and thermostat frequency $1/\tau_{\mathrm{bath}}=4$, the remaining parameters being those of the preamble of \cref{sec:experiments}. 

We regard this observed agreement as a strong consistency check of the kinetic theory: the critical temperature \cref{eq:Tc} was derived from the self-consistency equation alone, with no input from the simulation, and the numerical simulation of the kinetic equation reproduces it without adjustment. The continuous character of the transition agrees with the Landau--de Gennes theory of two-dimensional nematics. In two dimensions the $Q$-tensor is a $2\times 2$ symmetric traceless matrix, for which $\tr Q^3 = 0$, so the Landau expansion of the free energy contains no cubic invariant and the mean-field transition is of second order. In three dimensions the cubic invariant does not vanish, and the transition is of first order \rev{\cite[Chapter~2]{degennes}}. Our observation is also consistent with Fatkullin and Slastikov's analysis of the critical points of the Onsager functional \cite{fatkullinSlastikov} and with the nematic-MPCD simulations by Shendruk and Yeomans \cite{shendrukYeomans}, which also report a second-order transition in two dimensions.

\begin{figure}[tb]
	\centering
	\includegraphics[width=0.6\textwidth]{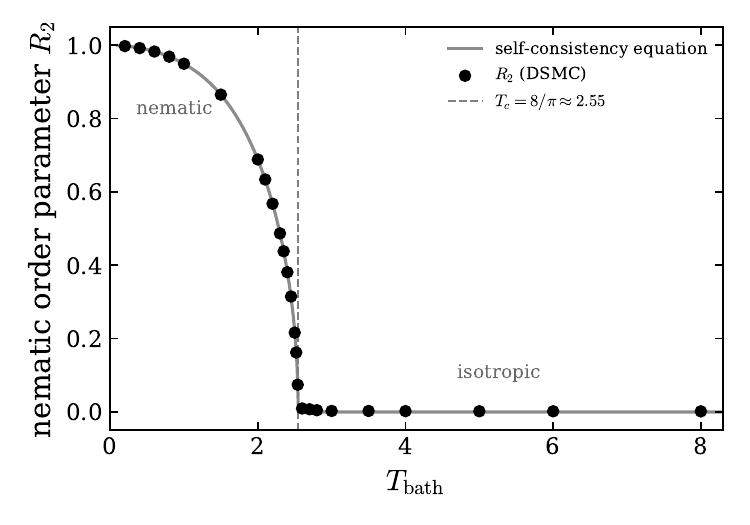}
	\caption{Test~4. \revjac{We simulate \eqref{eq:Boltzmann2D_calamitic} with the Onsager potential \cref{eq:onsager} and the thermostat of \cref{sec:thermostat} at a range of bath temperatures, to show the isotropic--nematic phase transition. Nematic order parameter $R_2=1-\sigma^2(\theta)$} of \cref{eq:circvar} at the final time as a function of the bath temperature $T_{\mathrm{bath}}$. The order parameter vanishes continuously at the critical temperature $T_c=8/\pi$ of \cref{prop:critical} (dashed line), the signature of a second-order transition. \revjac{At $T_c$ the isotropic branch $R_2=0$ loses its stability and the nematic branch bifurcates from it.} \reviv{One run per $T_{\mathrm{bath}}$ with $N=4\times10^5$, $\Delta t=0.05$, $1/\tau=1/\tau_{\mathrm{bath}}=4$, $L=\sqrt{12}$.}}
	\label{fig:phasediagram}
\end{figure}

\subsection{Test 5: The hard-needle kernel}
\label{sec:hardneedle}
The previous tests all employ a Max\-wellian kernel, where every molecular pair has an equal probability of colliding. To ensure our results are not artifacts of this simplification, we rerun the microcanonical Onsager test (\cref{sec:onsager}) using the physical hard-needle kernel (\cref{sec:kernels}), sampled via Bird's method. \revjac{The equation solved is \cref{eq:boltzmann_generic} with the collision operator \cref{eq:collision_operator_generic} and the hard-needle kernel, in place of the Maxwellian kernel of \cref{eq:Boltzmann2D_calamitic}.} As \cref{fig:hardneedle} shows, the core phenomenology perfectly matches the Maxwellian run (\cref{fig:onsager}). The rods still align into two head--tail symmetric nematic bumps, interaction energy converts into kinetic energy, and total energy remains strictly conserved up to $0.2\%$ fluctuations, with no visible drift. Notably, the hard-needle cross-section vanishes for perfectly parallel rods. Consequently, collisions become increasingly rare as the system orders. Despite this, the macroscopic relaxation remains unaltered over the timeframe of our simulation. %

\begin{figure}[tb]
	\centering
	\subfloat[$(\theta,\omega)$ density, $t=2.5$]{\includegraphics[width=0.32\textwidth]{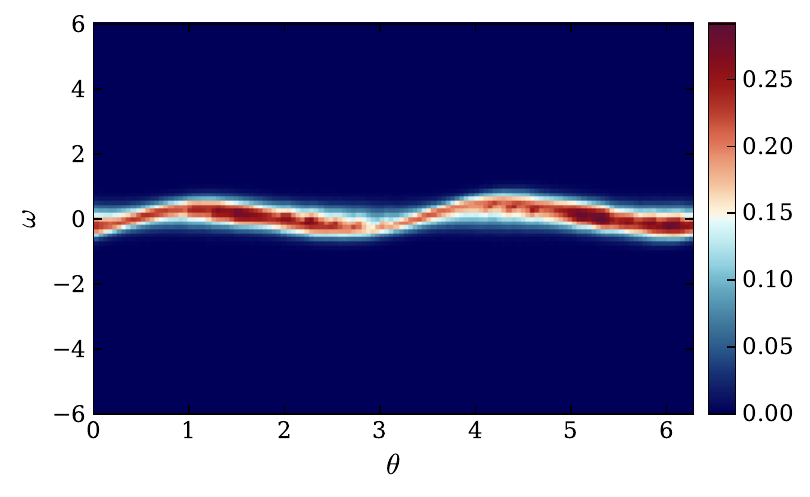}}\hfill
	\subfloat[$(\theta,\omega)$ density, $t=5$]{\includegraphics[width=0.32\textwidth]{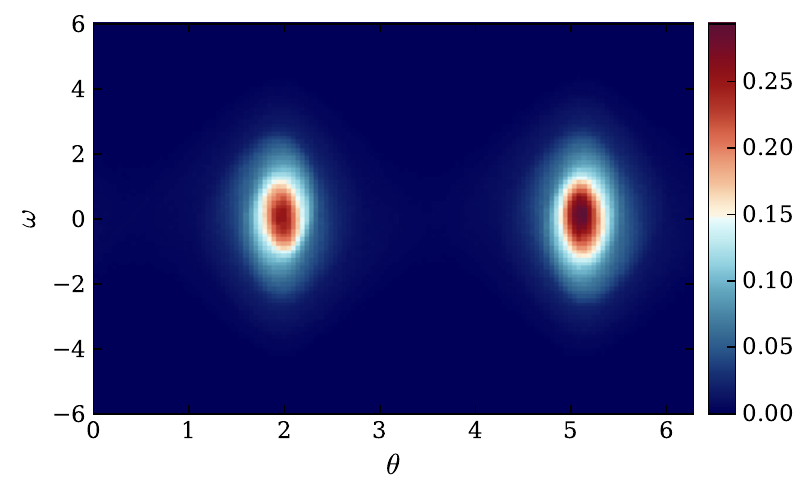}}\hfill
	\subfloat[$(\theta,\omega)$ density, $t=10$]{\includegraphics[width=0.32\textwidth]{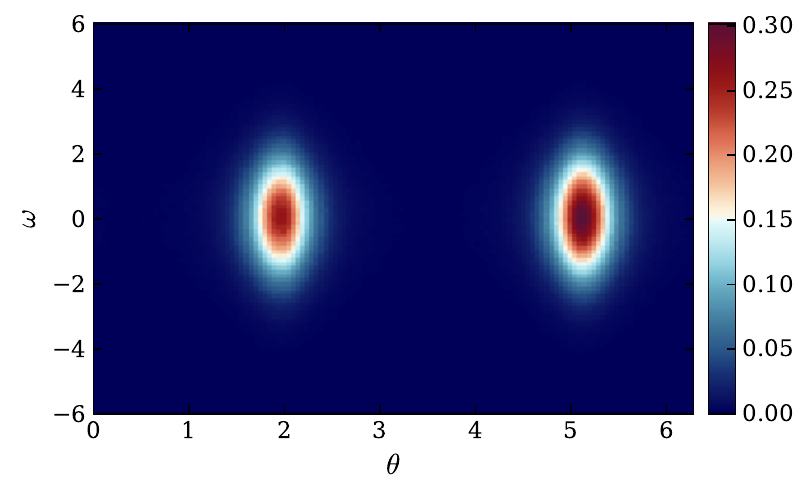}}\\
	\subfloat[temperature]{\includegraphics[width=0.24\textwidth]{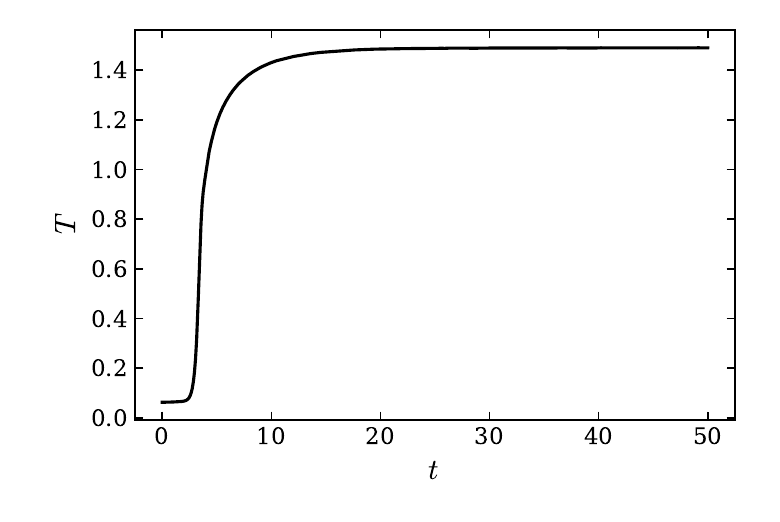}}\hfill
	\subfloat[circular variance]{\includegraphics[width=0.24\textwidth]{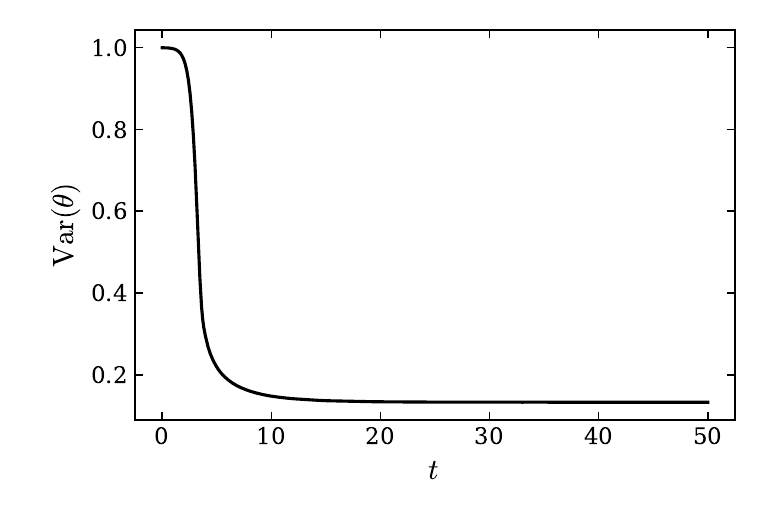}}\hfill
	\subfloat[interaction energy]{\includegraphics[width=0.24\textwidth]{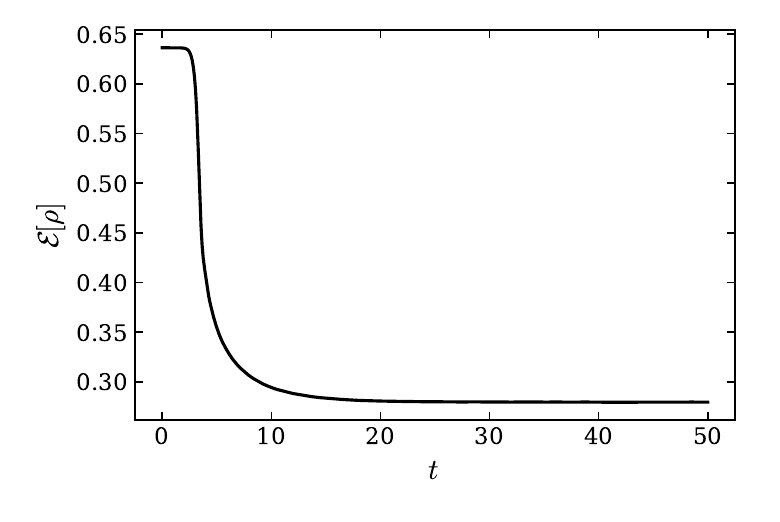}}\hfill
	\subfloat[total energy]{\includegraphics[width=0.24\textwidth]{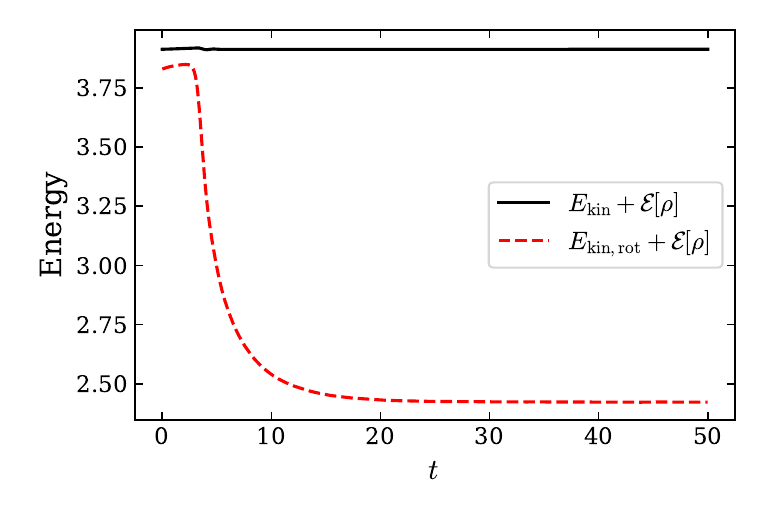}}
	\caption{Test~5. \revjac{We simulate \eqref{eq:boltzmann_generic} with the hard-needle kernel of \cref{sec:kernels}, sampled by Bird's method, and the Onsager potential \cref{eq:onsager} in the microcanonical ensemble, to show that the ordering dynamics is robust with respect to the choice of the collision kernel.} Layout as in \cref{fig:onsager}: three early snapshots of the $(\theta,\omega)$ density on top, the time evolution of the scalar observables below. The results match the Maxwellian-kernel run. In panel (g) the total energy is the solid line, conserved up to fluctuations of about $0.2\%$ of its value. \reviv{Parameters as in \cref{fig:onsager}. The $(\theta,\omega)$ densities are shown on $\omega\in[-6,6]$.}}
	\label{fig:hardneedle}
\end{figure}

\subsection{Test 6: Thermalization under a non-symmetric kernel}
\label{sec:handed}
\rev{Crucially, the $H$-theorem \cite{carrillo2025kinetic} relies on the reciprocity principle rather than strict symmetry of the transition probability. The handed exchange rule from \cref{sec:handedrule} perfectly distinguishes these two concepts: reciprocity holds exactly, but detailed balance fails pointwise. We therefore rerun the thermalization test from \cref{sec:thermalization} \revjac{for \cref{eq:Boltzmann2D_calamitic}} replacing the rigid-rod collision rule with this handed rule, keeping all other parameters and the initial datum \cref{eq:ic} identical. Because the initial datum is not equipartitioned ($T_{\mathrm{tr}}=1/12$ and $T_{\mathrm{rot}}=1/48$), total energy conservation determines the final equilibrium temperature at $T=(2\,T_{\mathrm{tr}}+T_{\mathrm{rot}})/3=1/16$.

\Cref{fig:handed} confirms our theoretical expectations. The velocity and angular-velocity marginals robustly relax to the Maxwellian. The translational and rotational temperatures cleanly converge to $1/16$ within just two time units and stay tightly coupled (within $0.2\%$ of each other) for the remainder of the simulation. Moreover, because the handed maps strictly exchange momentum and energy on the pair energy shell, the scheme conserves total energy and momentum to machine precision at every single collision. The simulation exhibits the behavior expected from the theory, even with the failure of detailed balance.}
\begin{figure}[tb]
	\centering
	\subfloat[$v_y$ marginal, $t=0$]{\includegraphics[width=0.32\textwidth]{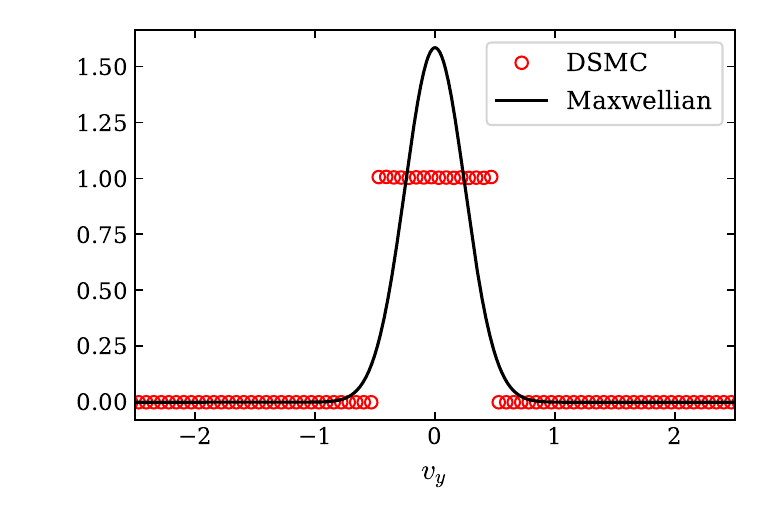}}\hfill
	\subfloat[$v_y$ marginal, $t=5$]{\includegraphics[width=0.32\textwidth]{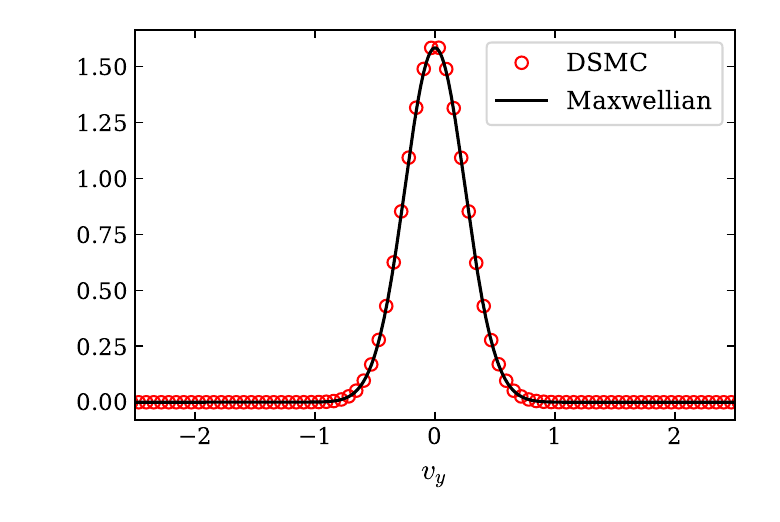}}\hfill
	\subfloat[$v_y$ marginal, $t=50$]{\includegraphics[width=0.32\textwidth]{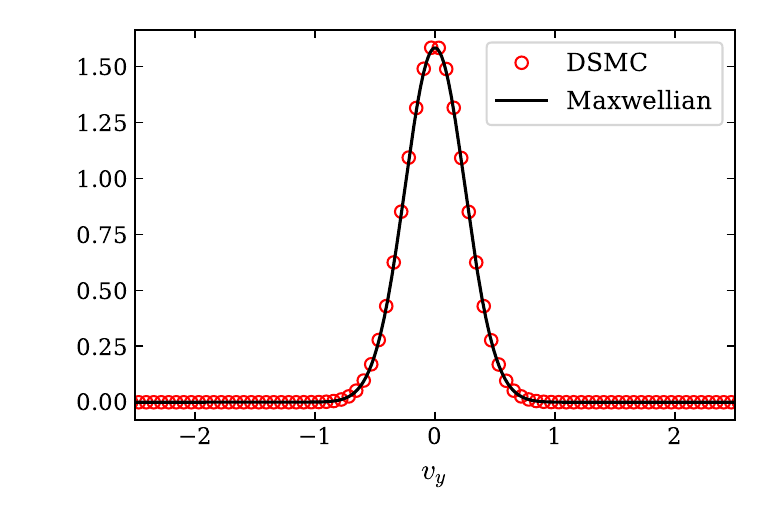}}\\
	\subfloat[$\omega$ marginal, $t=50$]{\includegraphics[width=0.32\textwidth]{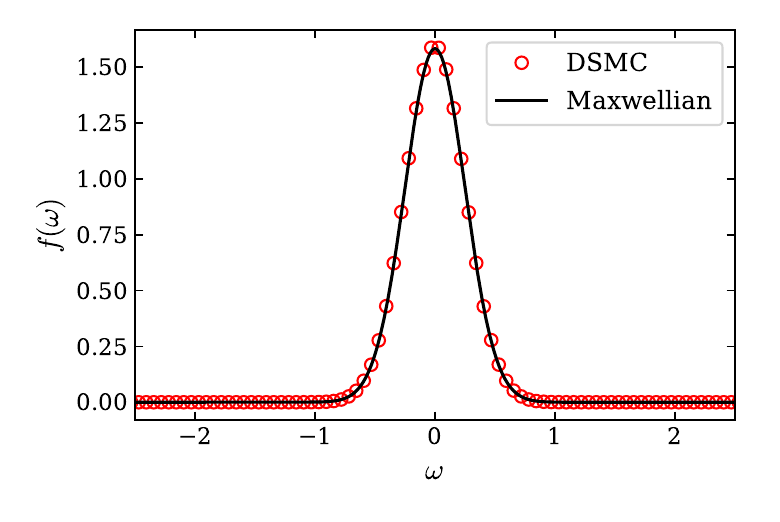}}\hfill
	\subfloat[temperatures]{\includegraphics[width=0.32\textwidth]{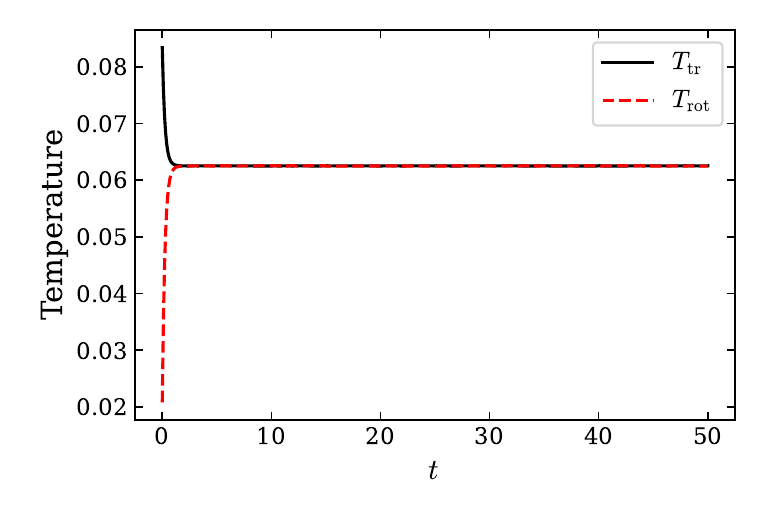}}\hfill
	\subfloat[temperature history]{\includegraphics[width=0.32\textwidth]{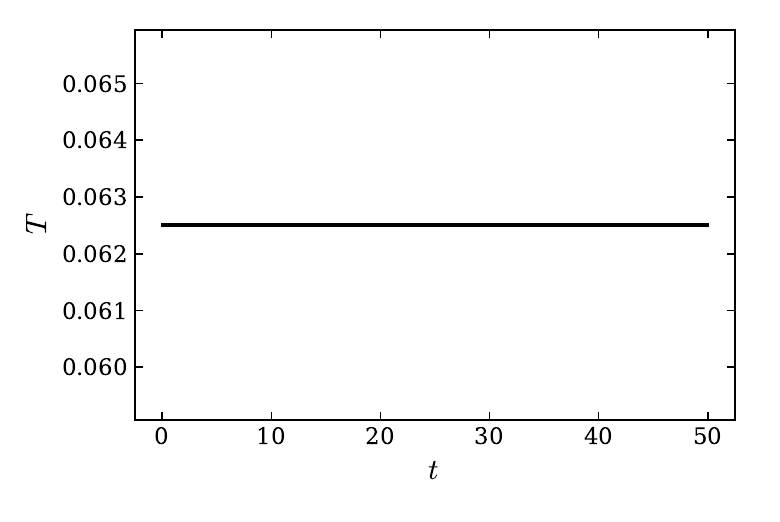}}\\
	\subfloat[$(v_x,v_y)$ density, $t=50$]{\includegraphics[width=0.32\textwidth]{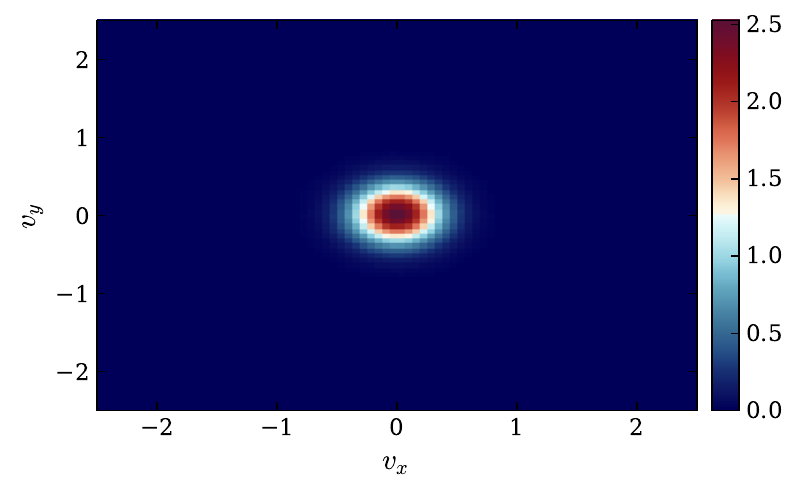}}\qquad
	\subfloat[$(\theta,\omega)$ density, $t=50$]{\includegraphics[width=0.32\textwidth]{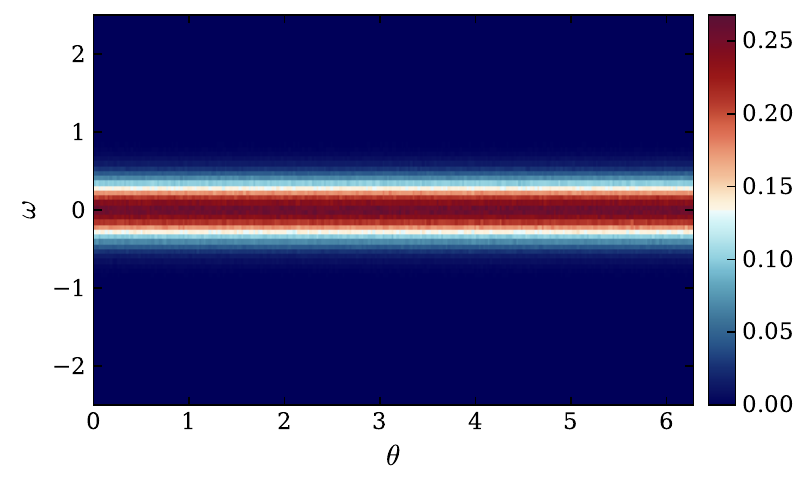}}
	\caption{Test~6. \revjac{We simulate \eqref{eq:Boltzmann2D_calamitic} with $\mathcal{V}=0$ and the handed exchange rule of \cref{sec:handedrule}, which violates detailed balance while satisfying the reciprocity principle, to show thermalization to the Maxwellian.} The top row shows three snapshots of the $v_y$ marginal, relaxing from the flat-top initial datum \cref{eq:ic} to the Maxwellian. The \revjac{second} row shows the $\omega$ marginal at the final time and the temperature histories. The translational and rotational temperatures converge to the value $T=1/16$ fixed by the conservation of the total energy. \revjac{The bottom row shows the $(v_x,v_y)$ and $(\theta,\omega)$ densities at the final time.} \reviv{$N=10^7$, $\Delta t=0.05$, $1/\tau=10$. The velocity and angular-velocity marginals are shown in detail on $[-2.5,2.5]$.}}
	\label{fig:handed}
\end{figure}

\section{Conclusion}
\label{sec:conclusion}

We have presented a direct simulation Monte Carlo scheme for the kinetic theory of ordered fluids \cite{carrillo2025kinetic}. The scheme is formulated for a generic order-parameter manifold and combines, through a Strang splitting, a Monte Carlo treatment of the binary collisions with a symplectic integration of the mean-field Vlasov transport. \revii{The collisions are treated by the Nanbu--Babovsky scheme for a Maxwellian kernel, by Bird's acceptance--rejection method for a general one, and by a BGK relaxation when they are replaced by a thermostat.} The mean-field force is reconstructed from the particle ensemble at every time step.

We have proved that the continuous dynamics conserves the total energy (\cref{prop:energy}), and we have verified that the numerical scheme approximately preserves this, up to a residual error below one percent over the full simulation horizon. This property is necessary to resolve the conversion of interaction energy into kinetic energy that accompanies the emergence of order.

The numerical tests confirm the relaxation to the Maxwellian predicted by the $H$-theorem and demonstrate the mean-field alignment induced by several interaction potentials, with both the Maxwellian and the hard-needle collision kernels. \rev{The relaxation persists under the handed exchange rule, whose transition probability violates detailed balance while satisfying the reciprocity principle, the regime for which the $H$-theorem was designed.} Varying the temperature with a BGK/Andersen thermostat, we recover the second-order isotropic--nematic transition at the critical temperature $T_c=8/\pi$ predicted by the linear stability analysis of the isotropic state and by the Landau--de Gennes theory of two-dimensional nematics.

\begin{remark}[Relation to multi-particle collision dynamics]
\label[remark]{rem:mpcd}
\revjac{The Andersen thermostat is also the central ingredient of multi-particle collision dynamics (MPCD), i.e. a mesoscopic method widely used for complex fluids, and particularly of the nematic-MPCD algorithm of Shendruk and Yeomans \cite{shendrukYeomans}. However, the two approaches discretize entirely different mathematical objects. Our scheme is a particle solver for the Boltzmann--Vlasov equation \cref{eq:boltzmann_generic} \cite{carrillo2025kinetic}. Its collisions are true binary interactions following microscopically derived scattering rules. The method explicitly approximates a kinetic equation whose $H$-theorem and Maxwellian equilibria are rigorous and known.}

\revjac{In contrast, MPCD evolves point particles through cell-based, multi-particle stochastic collisions. These collisions are phenomenologically designed to reproduce the correct macroscopic hydrodynamic limit, without relying on any specific underlying kinetic equation. In nematic-MPCD, alignment is enforced directly by relaxing the orientations in each cell toward the local director. In our scheme, alignment emerges organically from the mean-field force generated by the interaction potential. Finally, in MPCD, Andersen-type resampling is the defining mesoscopic collision mechanism. For us, it is merely an optional tool to fix the bath temperature.}
\end{remark}

Natural directions for future work include the extension to spatially inhomogeneous flows, to three-dimensional and head--tail symmetric molecules, and to asympto\-tic-preserving time-relaxed Monte Carlo schemes for the stiff collisional regime \cite{pareschi2001time}.

\section*{Declarations}
\revvii{The authors used generative AI tools in preparing this work. An initial DSMC code was written by the authors in PETSc \cite{petsc-user-ref}, building on earlier implementations in MATLAB. Claude Opus~5 and Claude Fable~5.1 (Anthropic) were later used to improve and extend this code, to write the scripts that launch the parameter sweeps and produce the figures from the stored simulation output, and to draft and edit the text. GPT 5.6 Sol, GPT-6 Astra (OpenAI) and Gemini~3.1 Pro (Google) were used to obtain critical readings of drafts of the manuscript. No AI tool was used in deriving or proving any of the mathematical results of the paper. The authors tested all code and checked every figure from data written by the simulations. The code and the scripts that reproduce the numerical tests of \cref{sec:experiments} are publicly available at \url{https://github.com/UZerbinati/dsmc}. The authors assume responsibility for all content.}

\bibliographystyle{abbrv}
\bibliography{references}
\end{document}

%% file: ex_shared.tex
\usepackage{lipsum}
\usepackage{amsfonts}
\usepackage{graphicx}
\usepackage{epstopdf}
\usepackage{algorithmic}
\usepackage{accents}
\usepackage{tikz}

\definecolor{verde}{RGB}{0,128,128}
\definecolor{rosso}{RGB}{191,0,64}
\tikzset{
  Arrow/.style={->, draw=black},
  GreenDashed/.style={-, draw=verde, thick, dashed},
  greenArrow/.style={->, draw=verde},
  ArrowRed/.style={->, draw=rosso},
  none/.style={inner sep=0pt, outer sep=0pt}
}
\pgfdeclarelayer{nodelayer}
\pgfdeclarelayer{edgelayer}
\pgfsetlayers{edgelayer,nodelayer,main}
\ifpdf
  \DeclareGraphicsExtensions{.eps,.pdf,.png,.jpg}
\else
  \DeclareGraphicsExtensions{.eps}
\fi

\newsiamremark{remark}{Remark}
\crefname{remark}{Remark}{Remarks}
\newsiamremark{hypothesis}{Hypothesis}
\crefname{hypothesis}{Hypothesis}{Hypotheses}
\newsiamthm{claim}{Claim}
\newsiamremark{fact}{Fact}
\crefname{fact}{Fact}{Facts}

\headers{Direct Simulation Monte Carlo for Ordered Fluids}{J.~A.~Carrillo, P.~E.~Farrell, A.~Medaglia, U.~Zerbinati}

\title{A Kinetic Theory Approach To Ordered Fluids: Direct Simulation Monte Carlo Methods\thanks{Submitted to the editors DATE.
\funding{This work was funded by 
the Engineering and Physical Sciences Research Council [grant number EP/W026163/1],
the Science and Technology Facilities Council [grant number UKRI/ST/B000495/1],
the Donatio Universitatis Carolinae Chair ``Mathematical modelling of multicomponent systems'',
and
the UKRI Digital Research Infrastructure Programme through the Science and Technology Facilities Council's Computational Science Centre for Research Communities (CoSeC).
The authors gratefully acknowledge the hospitality of the Erwin Schrödinger International Institute for Mathematics and Physics, Vienna, where part of this work was carried out. \\
JAC and AM acknowledge the support by the Advanced Grant Nonlocal-CPD,“Nonlocal PDEs for Complex Particle Dynamics: Phase Transitions, Patterns and Synchronization”, of the European Research Council Executive Agency (ERC) under the European Union’s Horizon 2020 research and innovation programme (grant 883363).
JAC was also partially supported by the EPSRC EP/V051121/1 and by the ``Maria de Maeztu'' Excellence Unit IMAG, reference CEX2020-001105-M, funded by MCIN/AEI/10.13039/501100011033/.
AM acknowledges the support by Fondo Italiano per la Scienza (FIS2023-01334) advanced grant “ADvanced numerical Approaches for MUltiscale Systems with uncertainties” - ADAMUS.
The research of U.Z.~was supported by the Research Council of Norway, project number 357556 (LaVa).
}}}

\author{
José A. Carrillo\thanks{Mathematical Institute, University of Oxford, Oxford, UK(\email{carrillo@maths.ox.ac.uk}).} 
\and Patrick E.~Farrell\thanks{Mathematical Institute, University of Oxford, Oxford, UK and Mathematical Institute, Faculty of Mathematics and Physics, Charles University, Czechia (\email{patrick.farrell@maths.ox.ac.uk}).}
\and Andrea Medaglia\thanks{Department of Mathematics and Computer Science, University of Ferrara, Ferrara, IT (\email{andrea.medaglia@unife.it}).}
\and Umberto Zerbinati\thanks{Simula Research Laboratory, Oslo, NO (\email{umberto@simula.no}).}
}

\usepackage{amsopn}

\usepackage{mathtools}
\usepackage{todonotes}
\makeatletter
\newcommand{\mat}[1]{{\mathpalette\mat@{#1}}}
\newcommand{\mat@}[2]{%
  \begingroup
  \sbox\z@{$\m@th#1\underline{#2}$}%
  \dimen@=\dp\z@ \advance\dimen@ -2\mat@dimen{#1}%
  \dp\z@=\dimen@
  \sbox\z@{$\m@th\underline{\box\z@}$}%
  \box\z@
  \endgroup
}
\newcommand\mat@dimen[1]{%
  \fontdimen8
  \ifx#1\displaystyle\textfont\else
  \ifx#1\textstyle\textfont\else
  \ifx#1\scriptstyle\scriptfont\else
  \scriptscriptfont\fi\fi\fi 3
}
\newcommand{\tensor}[1]{{\mathpalette\tensor@{#1}}}
\newcommand{\tensor@}[2]{%
  \begingroup
  \sbox\z@{$\m@th#1\underline{#2}$}%
  \dimen@=\dp\z@ \advance\dimen@ -2\tensor@dimen{#1}%
  \dp\z@=\dimen@
  \sbox\z@{$\m@th\underline{\box\z@}$}%
  \box\z@
  \endgroup
}
\newcommand\tensor@dimen[1]{%
  \fontdimen8
  \ifx#1\displaystyle\textfont\else
  \ifx#1\textstyle\textfont\else
  \ifx#1\scriptstyle\scriptfont\else
  \scriptscriptfont\fi\fi\fi 3
}
\def\contraction{\vbox{\baselineskip2.5\p@ \lineskiplimit\z@
  \kern\p@\hbox{.}\hbox{.}\hbox{.}\hbox{.}}}
\makeatother

\definecolor{revorange}{RGB}{217,95,2}
\definecolor{revviolet}{RGB}{117,45,160}
\definecolor{revochre}{RGB}{184,134,11}
\definecolor{revgreen}{RGB}{27,120,55}
\newcommand{\rev}[1]{#1}
\newcommand{\revii}[1]{#1}
\newcommand{\reviii}[1]{#1}
\newcommand{\reviv}[1]{#1}
\newcommand{\revjac}[1]{#1}
\newcommand{\revam}[1]{{#1}}
\newcommand{\revvi}[1]{#1}
\newcommand{\revvii}[1]{#1}

\DeclareMathOperator{\tr}{tr}
\let \vec \underline
\let \pt \mathbf
\newsiamremark{example}{Example}
\crefname{example}{Example}{Examples}
\makeatletter
\newcommand{\hathat}[1]{%
\begingroup%
  \let\macc@kerna\z@%
  \let\macc@kernb\z@%
  \let\macc@nucleus\@empty%
  \hat{\raisebox{.28ex}{\vphantom{\ensuremath{#1}}}\smash{\hat{#1}}}%
\endgroup%
}
\makeatother